\documentclass[11pt]{article}
\usepackage{fullpage}
\usepackage{subcaption}
\usepackage[T1]{fontenc}
\usepackage{lmodern}
\usepackage{amsmath}
\usepackage{amssymb}
\usepackage{amsthm}
\usepackage{color,soul}
\usepackage{xcolor}
\usepackage{graphicx}
\usepackage{tcolorbox}
\usepackage[normalem]{ulem}
\usepackage{float}
\usepackage{thmtools}
\usepackage{hyperref}
\usepackage{cleveref}
\usepackage{microtype}
\usepackage{caption}
\usepackage{bbm}
\hypersetup{colorlinks=true,citecolor=blue, linkcolor=blue, urlcolor=blue}
\usepackage[linesnumbered,boxed,ruled,vlined]{algorithm2e}
\usepackage{bm}
\usepackage[numbers]{natbib}
\usepackage{enumerate} 
\usepackage{enumitem}
\usepackage{tabularx}
\usepackage{array}
\usepackage[thinlines]{easytable}
\usepackage{tikz}

\usetikzlibrary{decorations.pathmorphing, arrows.meta}
\newcolumntype{L}[1]{>{\raggedright\arraybackslash}p{#1}}
\newcolumntype{C}[1]{>{\centering\arraybackslash}m{#1}}
\newcolumntype{R}[1]{>{\raggedleft\arraybackslash}p{#1}}

\usepackage{makecell}
\usepackage{footnote}
\makesavenoteenv{tabular}

\newcommand{\DTV}[2]{d_{\mathrm{TV}}\left({#1},{#2}\right)}

\renewcommand{\epsilon}{\varepsilon}

\newtheorem{theorem}{Theorem}[section]

\newtheorem*{claim*}{Claim}

\newtheorem{lemma}[theorem]{Lemma}

\newtheorem{corollary}[theorem]{Corollary}
\theoremstyle{definition}
\newtheorem{definition}[theorem]{Definition}
\newtheorem{problem}[theorem]{Problem}

\newtheorem*{remark*}{Remark}
\newtheorem{assumption}[theorem]{Assumption}

\newcommand{\abs}[1]{\left\vert#1\right\vert}

 \newcommand{\inner}[2]{\langle #1,#2\rangle}

\newcommand{\defeq}{\triangleq}

\def\*#1{\mathbf{#1}} % Use \*A for \mathbf{A}
\def\+#1{\mathcal{#1}} % Use \+A for \mathcal{A}
\def\-#1{\mathrm{#1}} % Use \-A for \mathrm{A}
\def\^#1{\mathbb{#1}} % Use \^A for \mathbb{A}

\usepackage{todonotes}
\usepackage{xifthen}

\renewcommand{\Pr}[2][]{ \ifthenelse{\isempty{#1}}
  {\mathbf{Pr}\left[#2\right]} {\mathbf{Pr}_{#1}\left[#2\right]} } % Optional argument gives a subscripted probability.
\newcommand{\E}[2][]{ \ifthenelse{\isempty{#1}}
  {\mathbf{\mathbf{E}}\left[#2\right]}
  {\mathbf{\mathbf{E}}_{#1}\left[#2\right]} }
  \newcommand{\Var}[2][]{ \ifthenelse{\isempty{#1}}
  {\mathbf{\mathbf{Var}}\left[#2\right]}
  {\mathbf{\mathbf{Var}}_{#1}\left[#2\right]} }

\crefname{theorem}{Theorem}{Theorems}
\crefname{observation}{Observation}{Observations}
\crefname{claim}{Claim}{Claims}
\crefname{condition}{Condition}{Conditions}
\crefname{algorithm}{Algorithm}{Algorithms}
\crefname{property}{Property}{Properties}
\crefname{example}{Example}{Examples}
\crefname{fact}{Fact}{Facts}
\crefname{lemma}{Lemma}{Lemmas}
\crefname{corollary}{Corollary}{Corollaries}
\crefname{definition}{Definition}{Definitions}
\crefname{remark}{Remark}{Remarks}
\crefname{proposition}{Proposition}{Propositions}
\crefname{equation}{equation}{equations}
\crefname{enumi}{}{}
\crefname{enumii}{}{}
\crefname{enumiii}{}{}
\crefname{enumiv}{}{}
\creflabelformat{enumi}{#2#1#3}
\creflabelformat{enumii}{#2#1#3}
\creflabelformat{enumiii}{#2#1#3}
\creflabelformat{enumiv}{#2#1#3}

\newboolean{DoubleBlind}
\setboolean{DoubleBlind}{false}

\title{On Deterministically Computing Total Variation Distance\\via Zonotope Compression}

\author{\ifthenelse{\boolean{DoubleBlind}}{Author(s)}
{Yucheng Fu\footnote{School of Computing and Data Science, The University of Hong Kong. \\\makebox[1.78em][l]{}Email: \texttt{fyc0130@connect.hku.hk}.}
}}

\date{}

\begin{document}

\maketitle
\begin{abstract}
  We study deterministic relative approximation of the total variation distance between high-dimensional distributions given by succinct descriptions. We develop an abstract deterministic approximation framework based on representing the total variation distance as a support function of a low-dimensional zonotope.
  As applications, we obtain FPTASs for several models. Given two mixtures of product distributions over $[q]^n$ with a total of $K$ component distributions, our algorithm approximates their TV-distance within a factor of $1+\varepsilon$ in time $\widetilde O_K(nq(n/\varepsilon)^{2K})$. We also give an FPTAS for mixtures of $n$-step Markov chains over $[q]^n$ with a total of $K$ component distributions, with running time $\widetilde O_K(nq^2(n/\varepsilon)^{2K})$. Finally, for two latent-tree Ising models with the same underlying tree topology, we give an FPTAS for the TV-distance between their leaf marginals in time $O(|V|^{13}\varepsilon^{-12})$.
\end{abstract}

\section{Introduction}
Let $\^P$ and $\^Q$ be two discrete distributions over the same space $\Omega$. The \emph{total variation distance (TV-distance)} between $\^P$ and $\^Q$ is defined as
\begin{align*}
    \DTV{\^P}{\^Q} = \frac{1}{2} \sum_{\omega \in \Omega} \left| \^P(\omega) - \^Q(\omega) \right|  =  \sum_{\omega \in \Omega}\max\{0, \^P(\omega) - \^Q(\omega)\}.
\end{align*}

The TV-distance is one of the most fundamental metrics for measuring the difference between distributions, as it characterizes the optimal distinguishability between $\^P$ and $\^Q$. Computing the TV-distance is therefore a basic problem arising in probability theory, statistics, and learning theory.

We study the following problem in this paper:

\begin{problem}[Relative approximation of TV-distance]\label{problem:tv-approx}
  Approximate the TV-distance between two distributions over the same finite sample space $\Omega$.
  \begin{itemize}
      \item \textbf{Input}: An error bound $\varepsilon \in (0,1)$ and descriptions of two distributions $\^P,\^Q$ over $\Omega$.
      \item \textbf{Output}: A number $\hat{d}$ such that $\hat{d} \leq \DTV{\^P}{\^Q} \leq (1+\varepsilon)\hat{d}$.
  \end{itemize}
\end{problem}

The difficulty is that in high-dimensional settings the sample space $\Omega$ is usually exponentially large, so the defining sum cannot be evaluated by brute force even when the distributions themselves have \emph{succinct descriptions}.

A line of research has investigated the problem of \emph{approximating} TV-distance between high-dimensional distributions, especially when the two distributions admit \emph{succinct representations}. Early works~\cite{sahai2003complete, CanonneR14,ChenK14,Kiefer18,BGMV20} studied the \emph{additive-error} approximation and related complexity questions. More recently, \cite{BGMMPV23} showed that even for two product distributions, the \emph{exact computation} of TV-distance is $\#\mathsf{P}$-hard, and initiated the study of efficient \emph{relative-error} approximation. Subsequently, relative approximation algorithms have been developed in both \emph{randomized} and \emph{deterministic} directions. Randomized approximation schemes are known for product distributions~\cite{FGJW23} and their mixtures~\cite{FengFYZ26,BCFM26}, Bayesian networks~\cite{BGMMPV24ICML}, and undirected graphical models~\cite{feng2025approximating}. Far fewer deterministic relative-approximation algorithms are known. They cover restricted product distributions~\cite{BGMMPV23}, general product distributions and Markov chains~\cite{FengLL24}, and multivariate Gaussians via a reduction to product-structured Gaussian distributions~\cite{BFS25}.

The deterministic algorithms of~\cite{FengLL24} are based on likelihood-ratio sparsification. As they noted, this technique relies on a strong conditional-independence property of the distribution, and it is natural to ask whether this restriction can be relaxed. Mixture models are a basic test case for this question. Although each component may have simple structure, the hidden component index induces a \emph{non-local dependency} that couples all variables simultaneously, and the probability mass becomes a linear combination of component probabilities.

\paragraph{Computational model}
Throughout the paper, running times are measured in an extended \emph{real-arithmetic model}, where comparison, addition, subtraction, multiplication, division, square root, and exponentiation take unit time, and the exact symmetric inverse square root of an $r\times r$ positive definite matrix is available as a primitive of cost $O(r^3)$. Under this model, all running-time bounds below count arithmetic operations. We do not analyze the bit lengths of intermediate values or the bit complexity of these operations. The notation $\widetilde O$ suppresses logarithmic factors in $n$, $q$, and $1/\varepsilon$.

In this paper, we develop an abstract \emph{deterministic framework} for this direction in \Cref{sec:abstract-framework}. Our main contribution is deterministic relative approximation for latent-variable models, including mixtures of product distributions (\Cref{sec:mix-of-prod}), mixtures of Markov chains (\Cref{sec:mix-of-mc}), and same-topology latent-tree Ising models (\Cref{sec:latent-tree}). These models have \emph{non-local dependencies} induced by hidden variables. We realize TV-distance as the support function of a low-dimensional represented zonotope whose generators encode local or component probability vectors, and recursively maintain a compressed representation that preserves all support functions up to relative error.

Our technique builds on classical geometric results. Approximation of zonoids by zonotopes with a small number of generators was studied by Bourgain, Lindenstrauss, and Milman \cite{BourgainLM89}, with further developments including \cite{Matousek96Zonotopes,ReisRothvoss26}. Aggregation-based compression and zonotope order-reduction methods have also been studied in convex geometry and control \cite{CampiHW94,YangO22,LutzowKA25}. Our normalization uses $\ell_1$ Lewis weights \cite{Lewis78} and the algorithmic treatment of Cohen and Peng \cite{CohenP15}. We combine these geometric ingredients with the recursive descriptions of the input distributions to obtain deterministic TV-distance approximation schemes.

Our first application is to \emph{mixtures of product distributions}. Mixtures of product distributions are a natural extension of product distributions. Their identity-testing problem has been studied in~\cite{0001GMMPV25}, and FPRASs for \Cref{problem:tv-approx} were given in~\cite{FengFYZ26,BCFM26}. Let $[q]=\{1,\ldots,q\}$ and let $\Omega=[q]^n$. Fix constants $k_1,k_2\geq1$. For each $s\in[k_1]$, let $\^P^{(s)}=\bigotimes_{i=1}^n\^P_i^{(s)}$ be a product distribution over $\Omega$, where $\^P_i^{(s)}$ denotes the marginal on the $i$-th coordinate. Let $\alpha^{(1)},\ldots,\alpha^{(k_1)}$ be mixing weights satisfying $\sum_{s=1}^{k_1}\alpha^{(s)}=1$. Thus the mixture $\^P$ is specified by the mixing weights and the coordinate marginals of all component product distributions, and is written as
$\^P=\sum_{s=1}^{k_1}\alpha^{(s)}\^P^{(s)}$.
More explicitly, for every $x\in\Omega$,
$$\^P(x)=\sum_{s=1}^{k_1}\alpha^{(s)}\^P^{(s)}(x)=\sum_{s=1}^{k_1}\alpha^{(s)}\prod_{i=1}^n\^P_i^{(s)}(x_i).$$
The distribution $\^Q$ is given analogously, with $k_2$ component product distributions $\^Q^{(1)},\ldots,\^Q^{(k_2)}$ and mixing weights $\beta^{(1)},\ldots,\beta^{(k_2)}$ satisfying $\sum_{t=1}^{k_2}\beta^{(t)}=1$. Mixtures of product distributions are a well-studied class of \emph{latent-variable models} in machine learning~\cite{FeldmanOS08,JainO14,ChenM19,GordonMRS21,GordonJMRS24}. To draw a random sample $X\sim\^P$, one first samples a \emph{hidden} component index $S\in[k_1]$ according to the mixing weights, and then samples $X\sim\^P^{(S)}$ from the selected product distribution. 
The total input description consists of the mixing weights and all one-coordinate marginals, and has size $O((k_1+k_2)nq)$ for the two mixtures. Our first result gives a deterministic relative-approximation algorithm for this model.

\begin{theorem}[TV-distance between mixtures of product distributions]\label{thm:mix-of-prod}
  Let $K=k_1+k_2$ be a constant. There exists an FPTAS that solves \Cref{problem:tv-approx} for mixtures of product distributions: given as input $\varepsilon \in (0,1)$, mixing weights $\alpha^{(1)},\ldots,\alpha^{(k_1)}$ and $\beta^{(1)},\ldots,\beta^{(k_2)}$, and one-coordinate marginals $\^P_i^{(s)}$ for all $s\in[k_1], i\in[n]$ and $\^Q_i^{(t)}$ for all $t\in[k_2], i\in[n]$, specifying two mixtures $\^P,\^Q$ over $[q]^n$, the algorithm outputs a number $\hat{d}$ such that $\hat{d} \leq \DTV{\^P}{\^Q} \leq (1+\varepsilon)\hat{d}$ in time $\widetilde O_K(nq(n/\varepsilon)^{2K})$.
\end{theorem}

\Cref{thm:mix-of-prod} provides a \emph{deterministic} relative-approximation algorithm for the TV-distance between mixtures of product distributions. 
The deterministic version of the non-mixture product case has been studied in \cite{FengLL24}, which gave a deterministic algorithm with running time
$O(\frac{qn^2}{\varepsilon}\log q\log\frac{n}{\varepsilon\DTV{\^P}{\^Q}})$.
For mixtures of product distributions, the identity testing problem, i.e.\ deciding whether $\DTV{\^P}{\^Q}=0$, was studied in \cite{0001GMMPV25}, which gave a deterministic algorithm with running time $\mathrm{poly}(n,q,K)$. For relative approximation, \cite[Theorem 1.2]{FengFYZ26} gives a randomized algorithm with running time $O(q^2(nq)^{2K-2}/\varepsilon^2)$ and success probability at least $99\%$, where $K=k_1+k_2$. Our algorithm is deterministic and runs in time $\widetilde O_K(nq(n/\varepsilon)^{2K})$. The dependence on $\varepsilon$ increases from $\varepsilon^{-2}$ to $\varepsilon^{-2K}$, up to logarithmic factors. The contribution is the deterministic relative guarantee for mixtures; both running-time bounds are polynomial in $n$, $q$, and $1/\varepsilon$ for fixed $K$.

We apply the abstract framework to the vector of component probabilities. The product structure provides a coordinatewise recursion, while the mixture weights enter only in the final linear combination. The dependence on $K=k_1+k_2$ appears in the exponent of the running time, so the theorem gives an FPTAS when $K=O(1)$.

Our second application is to \emph{mixtures of Markov chains}. Markov chains extend product distributions by allowing limited dependence between adjacent coordinates. Mixtures of Markov chains are therefore a natural extension of mixtures of product distributions. Let again $\Omega=[q]^n$. For each $s\in[k_1]$, let $\^P^{(s)}$ be an $n$-step Markov-chain distribution over $\Omega$, specified by an initial distribution $\^P_1^{(s)}$ over $[q]$ and transition matrices $\^P_{2|1}^{(s)},\ldots,\^P_{n|n-1}^{(s)}$. Let $\alpha^{(1)},\ldots,\alpha^{(k_1)}$ be mixing weights satisfying $\sum_{s=1}^{k_1}\alpha^{(s)}=1$. Thus the mixture $\^P$ is specified by the mixing weights and, for every component, its initial distribution and transition matrices; for every $x=(x_1,\ldots,x_n)\in\Omega$, it is given by
$$\^P(x)=\sum_{s=1}^{k_1}\alpha^{(s)}\^P^{(s)}(x)=\sum_{s=1}^{k_1}\alpha^{(s)}\^P_1^{(s)}(x_1)\prod_{i=2}^n\^P_{i|i-1}^{(s)}(x_i|x_{i-1}).$$
The distribution $\^Q$ is given analogously, with $k_2$ Markov-chain components $\^Q^{(1)},\ldots,\^Q^{(k_2)}$ and mixing weights $\beta^{(1)},\ldots,\beta^{(k_2)}$ satisfying $\sum_{t=1}^{k_2}\beta^{(t)}=1$. A mixture of Markov chains combines local temporal dependence inside each component with a global hidden choice shared by the entire trajectory, and has been studied in learning~\cite{GuptaKV16,SpaehT23,KausikTT23}.  The total input description consists of the mixing weights, initial distributions, and transition matrices, and has size $O((k_1+k_2)nq^2)$ for the two mixtures. Our second result gives a deterministic approximation scheme for this mixture-of-Markov-chains model.

\begin{theorem}[TV-distance between mixtures of Markov chains]\label{thm:mix-of-mc}
  Let $K = k_1 + k_2$ be a constant. There exists an FPTAS that solves \Cref{problem:tv-approx} for mixtures of $n$-step Markov-chain trajectories: given as input $\varepsilon \in (0,1)$, mixing weights $\alpha^{(1)},\ldots,\alpha^{(k_1)}$ and $\beta^{(1)},\ldots,\beta^{(k_2)}$, initial distributions $\^P_1^{(s)},\^Q_1^{(t)}$ and transition matrices $\^P_{i|i-1}^{(s)},\^Q_{i|i-1}^{(t)}$ for all $s\in[k_1]$, $t\in[k_2]$, and $i=2,\ldots,n$, specifying two mixtures of Markov chains $\^P,\^Q$ over $[q]^n$, the algorithm outputs a number $\hat{d}$ such that $\hat{d} \leq \DTV{\^P}{\^Q} \leq (1+\varepsilon)\hat{d}$ in time $\widetilde O_K(nq^2(n/\varepsilon)^{2K})$.
\end{theorem}

\Cref{thm:mix-of-mc} provides a relative-approximation algorithm for the TV-distance between mixtures of Markov chains. For Markov-chain distributions, \cite{FengLL24} studied the non-mixture Markov-chain case and gave a deterministic FPTAS for the TV-distance between two $n$-step Markov chains in time $O(\frac{q^2n^2}{\varepsilon}\log q\log \frac{n}{\varepsilon\DTV{\^P}{\^Q}})$. However, their setting has a single Markov chain on each side and does not include a latent component index. \Cref{thm:mix-of-mc} provides an extension of their result to mixtures of Markov chains.

Mixtures of Markov chains illustrate how the same framework handles local dependence inside each component. Unlike in the product case, the next coordinate is generated conditionally on the value of the previous one, so the recursive description has to keep track of the possible previous states. This is reflected in the additional $q$-dependence in the running time. 

By letting $k_1 = k_2 = 1$, we obtain FPTASs for the TV-distance between product distributions and between Markov chains as immediate corollaries of \Cref{thm:mix-of-prod,thm:mix-of-mc}. For both problems, the prior \emph{deterministic} relative-approximation algorithms of \cite{FengLL24} have running times that additionally depend on $\log({1}/{\DTV{\^P}{\^Q}})$. Our bounds remove this dependence. In the \emph{real-arithmetic model}, this gives strongly polynomial-time deterministic relative-approximation schemes, with arithmetic operation counts polynomial in $n$, $q$, and $1/\varepsilon$, and independent of the numerical value of $\DTV{\^P}{\^Q}$. Here strong polynomiality refers to this arithmetic operation count; it does not assert a bit-complexity bound.

\begin{corollary}[TV-distance between product distributions]
  There exists an FPTAS such that given any $\varepsilon \in (0,1)$ and two product distributions $\^P$, $\^Q$ over $[q]^n$, it outputs a real number $\hat{d}$, satisfying $\hat{d} \leq \DTV{\^P}{\^Q} \leq (1+\varepsilon)\hat{d}$, in time $\widetilde O\left(qn(n/\varepsilon)^4\right)$.
\end{corollary}

\begin{corollary}[TV-distance between Markov chains]
  There exists an FPTAS such that given any $\varepsilon \in (0,1)$ and two $n$-step Markov chains $\^P$, $\^Q$ over $[q]^n$, it outputs a real number $\hat{d}$, satisfying $\hat{d} \leq \DTV{\^P}{\^Q} \leq (1+\varepsilon)\hat{d}$, in time $\widetilde O\left(q^2n(n/\varepsilon)^4\right)$.
\end{corollary}

Our framework can also be applied to graphical models. The third application is to \emph{latent-tree Ising models}, a central family of tree-structured graphical models in which the observed variables are leaves and the internal variables are hidden. Let $T=(V,E)$ be an undirected tree with leaf set $L\subseteq V$. We consider two Ising models $(T,J^{\^P},h^{\^P})$ and $(T,J^{\^Q},h^{\^Q})$. Here $J^{\^P}$ and $J^{\^Q}$ are \emph{interaction matrices}, assigning one interaction value to each edge $\{u,v\}\in E$, and $h^{\^P},h^{\^Q}\in\^R^{V}$ are \emph{external field vectors}. The Ising model $(T,J^{\^P},h^{\^P})$ defines a Gibbs distribution $\^P$ as follows. For any $x_V\in\{\pm1\}^V$, define the \emph{weight function} $w_{\^P}$ and the \emph{partition function} $Z_{\^P}$ as
$$w_{\^P}(x_V)=\exp\left(\sum_{\{u,v\}\in E}J^{\^P}_{uv}x_ux_v+\sum_{u\in V}h^{\^P}_ux_u\right), \qquad Z_{\^P}=\sum_{x_V\in\{\pm1\}^V}w_{\^P}(x_V).$$
The full \emph{Gibbs distribution} $\^P$ induced by $(T,J^{\^P},h^{\^P})$ is given by $\^P(x_V)=w_{\^P}(x_V)/Z_{\^P}$. Let $\^P_L$ be the marginal distribution of $\^P$ over the leaves $L$, i.e.,
$$\^P_L(x_L)=\sum_{\substack{y\in\{\pm1\}^V:\\y_L=x_L}}\^P(y), \qquad x_L\in\{\pm1\}^L.$$
The distributions $\^Q$ and $\^Q_L$ are defined analogously, and the target is to approximate $\DTV{\^P_L}{\^Q_L}$. \emph{Same topology} means that the two models share the same tree $T$ and the same observed leaves, while their interaction matrices $J^{\^P}, J^{\^Q}$ and external fields $h^{\^P}, h^{\^Q}$ may be different. The input description consists of the tree and the two parameter sets $(J^{\^P},h^{\^P})$ and $(J^{\^Q},h^{\^Q})$. The hidden internal spins make each leaf probability a sum over many internal configurations, so the observed distribution has \emph{non-local dependencies} even though the underlying model is tree-structured. Such latent-tree models are widely studied in the learning literature~\cite{BreslerK20,DaganKD22,KandirosDDC23}. Relative-error randomized approximation algorithms are known for broader graphical models with parameter restrictions~\cite{feng2025approximating}; our third result gives a deterministic FPTAS for the same-topology latent-tree Ising models.

\begin{theorem}[Same-topology latent-tree Ising]\label{thm:latent-tree}
  There exists an FPTAS that solves \Cref{problem:tv-approx} for same-topology latent-tree Ising models: given as input $\varepsilon\in(0,1)$, a tree $T=(V,E)$ with leaf set $L$, edge interactions $J^{\^P},J^{\^Q}$, and external fields $h^{\^P},h^{\^Q}$, specifying two Ising models $(T,J^{\^P},h^{\^P})$ and $(T,J^{\^Q},h^{\^Q})$ whose leaf marginals are $\^P_L$ and $\^Q_L$, the algorithm outputs a number $\hat{d}$ such that $\hat{d} \leq \DTV{\^P_L}{\^Q_L} \leq (1+\varepsilon)\hat{d}$ in time $O\left(|V|(|V|/\varepsilon)^{12}\right)$.
\end{theorem}

\Cref{thm:latent-tree} provides a \emph{deterministic} relative-approximation algorithm for the TV-distance between latent-tree Ising models. Prior relative-approximation algorithms that apply to graphical models are \emph{randomized} and require parameter restrictions such as bounded interaction strengths~\cite{feng2025approximating}; by contrast, our algorithm is deterministic and imposes no restriction on the interactions or external fields, at the price of the same-topology assumption. The same-topology assumption is used to compare the two models through a common recursive tree structure. It also leaves open the more structural question of comparing latent-tree models with different underlying trees, which is closely related to learning and testing latent-tree Ising models~\cite{KandirosDDC23}. 

\paragraph{Subsequent work}
In subsequent work, Bhattacharyya, Cormode, Fu, and Meel~\cite[Theorem 1.1]{BCFM26} give an FPRAS for the TV-distance between mixtures of product distributions, using $\mathrm{poly}(n,q,K,1/\varepsilon,\log(1/\delta))$ arithmetic operations to obtain a $(1\pm\varepsilon)$ relative approximation with probability at least $1-\delta$. Their algorithm has polynomial dependence on the number of components $K$, while our algorithm provides a deterministic relative approximation for every fixed $K$. Both approaches preserve sums of absolute linear projections throughout the recursion. Their method uses randomized Lewis-weight sampling, while ours deterministically aggregates generators with nearby directions. Their work also extends to smooth, structured-decomposable probabilistic circuits that share a common v-tree.

\paragraph{Open problems}
Several natural directions remain to be explored. First, can the polynomial-size $\ell_1$ coresets used in~\cite{BCFM26} be constructed deterministically in polynomial time? Replacing our compression step with such a construction would yield deterministic TV-distance approximation schemes for mixtures of product distributions and Markov chains with running time $\mathrm{poly}(n,q,K,1/\varepsilon)$.

Second, it would be interesting to extend the framework to broader models. On the discrete side, hidden Markov models are a natural target, although TV-distance computation for general hidden Markov models is already known to have strong hardness and decidability barriers~\cite{Kiefer18}. Other possible targets include mixtures of Markov decision processes and related latent-context models, which have been studied from the learning perspective~\cite{KausikTT23,KwonECM23}. It would also be interesting to remove the same-topology assumption for latent-tree Ising models as stated after \Cref{thm:latent-tree}. Moreover, our results in this paper are for discrete domains, and another direction is to understand whether the represented-zonotope viewpoint can be adapted to continuous distributions. Natural candidates include multivariate Gaussians~\cite{BFS25}, mixtures of Gaussians~\cite{LiS17,AshtianiBHLPM18}, and Gaussian latent-tree models~\cite{DaganKD22}.

\section{Technical Overview}
We explain the main ideas through mixtures of product distributions. TV-distance becomes the support function of a low-dimensional zonotope, and the product structure gives a recursion for this zonotope. A geometric compression step keeps its representation small while preserving the support function up to a relative factor.

\subsection{From TV-distance to Represented Zonotopes}
Let $\^P=\sum_{s=1}^{k_1}\alpha^{(s)}\^P^{(s)}$ and $\^Q=\sum_{t=1}^{k_2}\beta^{(t)}\^Q^{(t)}$ be two mixtures over $\Omega=[q]^n$, and set $K=k_1+k_2$. For each outcome $x$, collect its probabilities under all component distributions into a vector
$$
  g(x)=\left(\^P^{(1)}(x),\ldots,\^P^{(k_1)}(x),
  \^Q^{(1)}(x),\ldots,\^Q^{(k_2)}(x)\right),
$$
and put the signed mixture weights in
$$
  \gamma=\left(\alpha^{(1)},\ldots,\alpha^{(k_1)},
  -\beta^{(1)},\ldots,-\beta^{(k_2)}\right).
$$
Then $\inner{\gamma}{g(x)}=\^P(x)-\^Q(x)$. We call such an encoding a \emph{TV-realization}. The root multiset $\+G^{\mathrm{root}}=\{g(x):x\in\Omega\}$ lies in dimension $K$, although it contains $q^n$ vectors.

The \emph{represented zonotope} generated by a finite multiset $\+G$ is
$$
  \+Z(\+G)=\left\{\sum_{g\in\+G}t_g g:t_g\in[-1,1]\right\}
  \defeq \sum_{g\in\+G}[-g,g].
$$
Its support function measures its extent in a given direction:
$$
  \sigma_{\+Z(\+G)}(\gamma)
  =\max_{y\in\+Z(\+G)}\inner{\gamma}{y}
  =\sum_{g\in\+G}\abs{\inner{\gamma}{g}}.
$$
For the root zonotope, this gives
$$
  2\DTV{\^P}{\^Q}
  =\sum_{x\in\Omega}\abs{\^P(x)-\^Q(x)}
  =\sigma_{\+Z(\+G^{\mathrm{root}})}(\gamma).
$$
\Cref{fig:tv-support-geometry} illustrates this identity for $\+Z=\+Z(\+G^{\mathrm{root}})$. Measuring the horizontal coordinate by $\inner{\gamma}{y}$, the furthest extent of $\+Z$ in direction $\gamma$ is $\sigma_{\+Z}(\gamma)$. The blue segment runs from zero to this value, and TV-distance is one half of it.

\begin{figure}[htbp]
  \centering
  \begin{tikzpicture}[
    x=.95cm,y=.95cm,
    font=\small,
    direction/.style={-{Latex[length=1.8mm]},line width=.6pt},
    support/.style={blue!65!black,line width=1.2pt,-{Latex[length=2mm]}}
  ]
    \path[draw=black!55,fill=blue!6,line width=.7pt]
      (-2.7,-.9)--(-1.3,-1.8)--(1.2,-.9)
      --(2.7,.9)--(1.3,1.8)--(-1.2,.9)--cycle;
    \node at (.25,.8) {$\+Z=\+Z(\+G^{\mathrm{root}})$};
    \draw[direction,black!45] (-3.1,0)--(4,0)
      node[right,black] {$\inner{\gamma}{y}$};
    \node at (3.65,.3) {$\gamma$};
    \draw[black!55,dashed] (2.7,-1.45)--(2.7,2.05);
    \node[anchor=west] at (2.85,1.65)
      {$\inner{\gamma}{y^*}=\sigma_{\+Z}(\gamma)$};
    \draw[blue!65!black,densely dotted,line width=.8pt]
      (2.7,.9)--(2.7,0);
    \draw[black!55,line width=.4pt]
      (2.56,0)--(2.56,.14)--(2.7,.14);
    \draw[support] (0,0)--(2.7,0)
      node[pos=.3,below=5pt] {$\sigma_{\+Z}(\gamma)$};
    \fill (0,0) circle (1.3pt) node[above left] {$0$};
    \fill (2.7,.9) circle (1.3pt) node[above left] {$y^*$};
    \fill[blue!65!black] (2.7,0) circle (1.3pt);
    \node at (.7,-2.35)
      {$\displaystyle \DTV{\^P}{\^Q}
        =\frac12\sigma_{\+Z}(\gamma)$};
  \end{tikzpicture}
  \caption{The blue segment represents $\sigma_{\+Z}(\gamma)=2\DTV{\^P}{\^Q}$,
  the extent of the root zonotope in direction $\gamma$.}
  \label{fig:tv-support-geometry}
\end{figure}
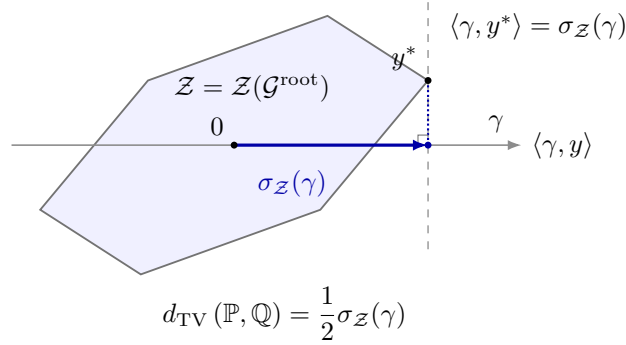

The support function is easy to evaluate from an explicit representation. The obstacle is that the exact root zonotope has an exponentially large representation.

\subsection{A Recursion for the Zonotope}
Let $\+Z_i$ be the zonotope associated with the first $i$ coordinates of the product mixtures. The initial zonotope is $\+Z_0=[-\mathbf 1,\mathbf 1]$, where $\mathbf 1\in\mathbb R^K$ is the all-ones vector, and the final zonotope is $\+Z_n=\+Z(\+G^{\mathrm{root}})$. For coordinate $i$ and symbol $v\in[q]$, define the diagonal map
$$
  D_{i,v}=\operatorname{diag}\left(
    \^P_i^{(1)}(v),\ldots,\^P_i^{(k_1)}(v),
    \^Q_i^{(1)}(v),\ldots,\^Q_i^{(k_2)}(v)\right).
$$
The product structure gives the geometric recurrence
$$
  \+Z_i=\sum_{v\in[q]}D_{i,v}\+Z_{i-1},
  \qquad i=1,\ldots,n,
$$
where the sum is a Minkowski sum. Each term accounts for one possible symbol at coordinate $i$. \Cref{fig:coordinate-circuit} displays this update as a circuit of linear maps and sums.

\begin{figure}[htbp]
  \centering
  \begin{tikzpicture}[
    x=.95cm,y=.9cm,
    font=\small,
    wire/.style={-{Latex[length=1.7mm]},draw=black!65,line width=.65pt},
    gate/.style={circle,draw=black!65,fill=blue!6,line width=.7pt,
      minimum size=7mm,inner sep=0pt},
    factor/.style={draw=black!45,fill=black!3,line width=.6pt,
      minimum width=10mm,minimum height=6mm,inner sep=3pt}
  ]
    \node[factor,draw=blue!65!black,fill=blue!6] (input) at (0,0)
      {$\+Z_{i-1}$};
    \node[gate] (timesone) at (-3.8,1.65) {$\times$};
    \node[gate] (timestwo) at (0,1.65) {$\times$};
    \node[gate] (timeslast) at (3.8,1.65) {$\times$};
    \node[factor] (mapone) at (-5.1,1.65) {$D_{i,1}$};
    \node[factor] (maptwo) at (-1.3,1.65) {$D_{i,2}$};
    \node[factor] (maplast) at (2.5,1.65) {$D_{i,q}$};
    \node at (1.15,1.65) {$\cdots$};
    \node[gate,draw=blue!65!black] (sum) at (0,3.2) {$+$};
    \node[above=5pt] at (sum.north) {$\+Z_i$};

    % A shared input feeds every diagonal-scaling branch.
    \draw[wire,blue!65!black] (input.north)--(timesone.south);
    \draw[wire,blue!65!black] (input.north)--(timestwo.south);
    \draw[wire,blue!65!black] (input.north)--(timeslast.south);
    \draw[wire] (mapone.east)--(timesone.west);
    \draw[wire] (maptwo.east)--(timestwo.west);
    \draw[wire] (maplast.east)--(timeslast.west);
    \draw[wire] (timesone.north)--(sum.south west);
    \draw[wire] (timestwo.north)--(sum.south);
    \draw[wire] (timeslast.north)--(sum.south east);
  \end{tikzpicture}
  \caption{The circuit applies the diagonal maps $D_{i,v}$ at its $\times$
  nodes and combines their outputs by a Minkowski sum at the $+$ node.}
  \label{fig:coordinate-circuit}
\end{figure}
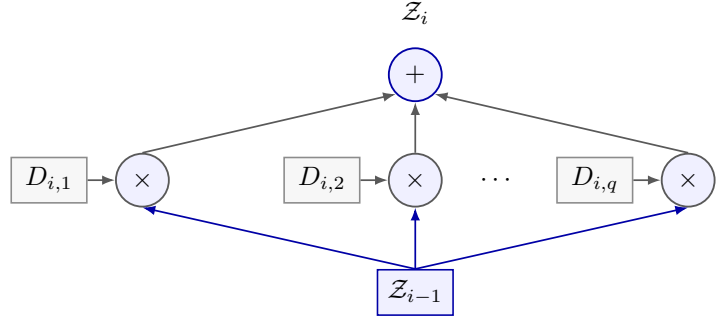

Both operations preserve set inclusion and a common multiplicative approximation factor. Consequently, replacing the input zonotope by a relative approximation gives an output with the same relative guarantee. The only additional error comes from compressing that output.

\subsection{Compressing the Recursion}
The compression lemma replaces a represented zonotope $\+Z$ by a zonotope $\underline{\+Z}$ with only $O_K(\eta^{-K})$ generators, satisfying
$$
  \underline{\+Z}\subseteq \+Z\subseteq(1+\eta)\underline{\+Z}.
$$
For fixed dimension $K$, this approximation can be computed deterministically in time polynomial in the representation size and $1/\eta$; see \Cref{lem:compression}.

As illustrated in \Cref{fig:compression-relative-support}, expanding $\underline{\+Z}$ about the origin by a factor of $1+\eta$ covers $\+Z$. Consequently, every support value has the same relative guarantee:
$$
  \sigma_{\underline{\+Z}}(\gamma)\le\sigma_{\+Z}(\gamma)
  \le(1+\eta)\,\sigma_{\underline{\+Z}}(\gamma).
$$
This is the geometric guarantee needed for a relative approximation of TV-distance.

\begin{figure}[htbp]
  \centering
  \begin{tikzpicture}[
    x=.85cm,y=.85cm,
    font=\small,
    direction/.style={-{Latex[length=1.8mm]},line width=.6pt},
    exact/.style={black!65,line width=.8pt},
    compressed/.style={blue!65!black,line width=.9pt},
    expanded/.style={blue!65!black,dashed,line width=.8pt},
    extent/.style={-{Latex[length=1.8mm]},line width=1pt}
  ]
    % The three outlines satisfy underline Z subset Z subset 1.4 underline Z.
    \begin{scope}[scale=1.4]
      \draw[expanded]
        (-2.235,-.81)--(-1.115,-1.53)--(.885,-.81)
        --(2.235,.81)--(1.115,1.53)--(-.885,.81)--cycle;
    \end{scope}
    \path[draw=black!65,fill=black!5,line width=.8pt]
      (-2.7,-.9)--(-1.3,-1.8)--(1.2,-.9)
      --(2.7,.9)--(1.3,1.8)--(-1.2,.9)--cycle;
    \path[draw=blue!65!black,fill=blue!12,line width=.9pt]
      (-2.235,-.81)--(-1.115,-1.53)--(.885,-.81)
      --(2.235,.81)--(1.115,1.53)--(-.885,.81)--cycle;
    \draw[direction,black!45] (-3.5,0)--(4.2,0)
      node[right,black] {$\inner{\gamma}{y}$};
    \node at (3.85,.3) {$\gamma$};
    \fill (0,0) circle (1.3pt) node[above left] {$0$};

    \draw[exact] (3.8,2.05)--(4.35,2.05)
      node[right,black] {$\+Z$};
    \draw[compressed] (3.8,1.5)--(4.35,1.5)
      node[right] {$\underline{\+Z}$};
    \draw[expanded] (3.8,.95)--(4.35,.95)
      node[right] {$(1+\eta)\underline{\+Z}$};

    \fill[blue!65!black] (2.235,.81) circle (1.2pt);
    \fill[black!65] (2.7,.9) circle (1.2pt);
    \fill[blue!65!black] (3.129,1.134) circle (1.2pt);
    \draw[blue!45,densely dotted] (2.235,.81)--(2.235,-2.65);
    \draw[black!40,densely dotted] (2.7,.9)--(2.7,-3.2);
    \draw[blue!45,densely dotted] (3.129,1.134)--(3.129,-3.75);
    \draw[black!25,densely dotted] (0,0)--(0,-3.75);

    \draw[extent,blue!65!black] (0,-2.65)--(2.235,-2.65);
    \node[anchor=west,blue!65!black] at (3.8,-2.65)
      {$\sigma_{\underline{\+Z}}(\gamma)$};
    \draw[extent,black!65] (0,-3.2)--(2.7,-3.2);
    \node[anchor=west] at (3.8,-3.2) {$\sigma_{\+Z}(\gamma)$};
    \draw[extent,blue!65!black,dashed] (0,-3.75)--(3.129,-3.75);
    \node[anchor=west,blue!65!black] at (3.8,-3.75)
      {$(1+\eta)\sigma_{\underline{\+Z}}(\gamma)$};
    \node at (1.25,-4.5)
      {$\displaystyle
        \sigma_{\underline{\+Z}}(\gamma)\le\sigma_{\+Z}(\gamma)
        \le(1+\eta)\sigma_{\underline{\+Z}}(\gamma)$};
  \end{tikzpicture}
  \caption{Compression preserves support values within a factor of $1+\eta$
  by maintaining $\underline{\+Z}\subseteq \+Z\subseteq(1+\eta)\underline{\+Z}$.}
  \label{fig:compression-relative-support}
\end{figure}
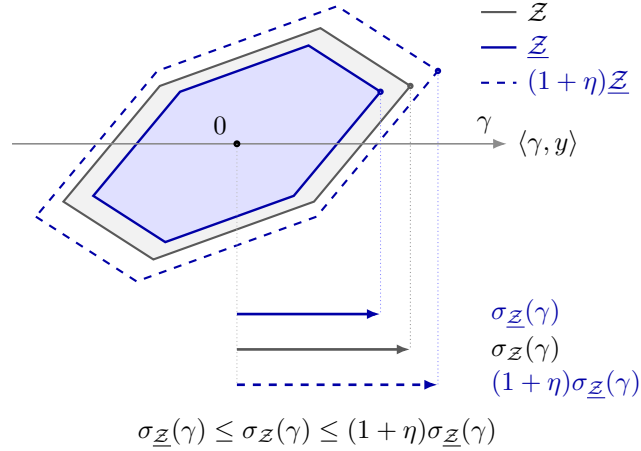

Returning to product mixtures, start with $\underline{\+Z}_0=\+Z_0$. At each coordinate, evaluate the geometric recurrence on the current approximation and compress its output. The linear maps and Minkowski sum preserve the existing approximation factor, while compression multiplies it by at most $1+\eta$. After $n$ steps,
$$
  \underline{\+Z}_n\subseteq \+Z_n
  \subseteq(1+\eta)^n\underline{\+Z}_n.
$$
Choose $\eta=(1+\varepsilon)^{1/n}-1$ and return $\hat{d}=\frac12\sigma_{\underline{\+Z}_n}(\gamma)$. By monotonicity of support functions, the output satisfies
$$
  \hat{d}\le\DTV{\^P}{\^Q}\le(1+\varepsilon)\hat{d}.
$$
Each compressed representation has size $O_K((n/\varepsilon)^K)$, so every step takes polynomial time when $K$ is fixed.

\Cref{sec:abstract-framework} formalizes this approach using local operations called \emph{gates}. In our framework, each gate must preserve relative zonotope approximations with a controlled accumulation of error. Product mixtures use the linear maps and Minkowski sums above. Mixtures of Markov chains additionally keep track of a boundary state, while the latent-tree application combines contributions from different subtrees. In each case, the algorithm alternates these operations with compression and evaluates the final support function to recover TV-distance. Subsequent work~\cite{BCFM26} further formalizes this approach in the setting of probabilistic circuits.

\section{The Abstract Framework}\label{sec:abstract-framework}
In this section, we set up an abstract framework for deterministically approximating the TV-distance between $\^P$ and $\^Q$ by a represented zonotope. First, we encode the signed mass $\^P(x)-\^Q(x)$ by low-dimensional generators. Second, we introduce a compression procedure to reduce the number of generators while preserving the corresponding support functions up to a relative factor. Third, we give an abstract FPTAS for the TV-distance problem when $\^P$ and $\^Q$ satisfy a certain assumption.

\subsection{TV-Distance Realization via Support Functions}

\begin{definition}\label{def:tv-realization}
For probability distributions $\^P$ and $\^Q$ over the finite space $\Omega=\{x_1,\ldots,x_{\abs{\Omega}}\}$, a TV-realization consists of a constant positive integer $K$, a known vector $\gamma\in\mathbb R^K$, and a multiset
$$ \mathcal G^{\mathrm{root}}=\{g(x):x\in\Omega\} \subseteq \^R^{K}, $$
such that 
$$ \inner{\gamma}{g(x)}=\^P(x)-\^Q(x) \qquad\text{for all }x\in{\Omega}. $$
\end{definition}

For discrete distributions, a TV-realization always exists, for instance by
taking $K=2$, $\gamma=(1,-1)$, and $g(x)=(\^P(x),\^Q(x))$.  The choice of
realization is not unique; in applications we use coordinates that expose the
recursive structure of the model.
%For discrete distributions, the representation in \Cref{def:tv-realization} is not restrictive. Indeed, for any pair $\^P$, $\^Q$ over $\Omega$, one may take $K=2$, $\gamma = (1,-1)$, and $g(x) = (\^P(x),\^Q(x))$. Thus $\inner{\gamma}{g(x)}=\^P(x)-\^Q(x)$. This realization is not unique. Different choices of coordinates may encode different structural information about the distribution. For example, suppose $\^P = \sum_{r=1}^{k_1}\alpha^{(r)}\^P^{(r)}$ and $\^Q = \sum_{r=1}^{k_2}\beta^{(r)}\^Q^{(r)}$ are two mixture models. Then a useful realization has $K=k_1+k_2$, $\gamma=(\alpha^{(1)},\ldots,\alpha^{(k_1)},-\beta^{(1)},\ldots,-\beta^{(k_2)})$, and $g(x) = (\^P^{(1)}(x), \ldots, \^P^{(k_1)}(x), \^Q^{(1)}(x), \ldots, \^Q^{(k_2)}(x))$. This realization separates the mixture weights from the component probabilities, which is what allows the component structure to be exposed from the mixture models.
%For general distributions $\^P$ and $\^Q$, a trivial approach for \Cref{def:tv-realization} is to let $K=2$, $\gamma = (1, -1)$, and $g(x) = (\^P(x), \^Q(x))$. However, this representation is useful when $\^P$ and $\^Q$ are mixtures of distributions. Formally, let $K=k_1+k_2$, $\^P = \sum_{i=1}^{k_1} \alpha^{(i)} \^P^{(i)}$ and $\^Q = \sum_{j=1}^{k_2} \beta^{(j)} \^Q^{(j)}$, where $\alpha^{(i)}, \beta^{(j)} \geq 0$ and $\sum_{i=1}^{k_1} \alpha^{(i)} = \sum_{j=1}^{k_2} \beta^{(j)} = 1$. Then we can let $\gamma = (\alpha^{(1)}, \ldots, \alpha^{(k_1)}, -\beta^{(1)}, \ldots, -\beta^{(k_2)})$ and $g(x) = (\^P^{(1)}(x), \ldots, \^P^{(k_1)}(x), \^Q^{(1)}(x), \ldots, \^Q^{(k_2)}(x))$.

\begin{definition}[Represented Zonotope]
For a positive integer $K$, and a finite multiset $\mathcal G = \{g_1, \ldots, g_m\} \subseteq\mathbb R^K$, the represented zonotope $\+Z(\mathcal G)$ is the centrally symmetric zonotope
$$ \+Z(\mathcal G) = \left\{\sum_{i=1}^m t_i g_i: t_i\in[-1,1]\right\} \defeq \sum_{i=1}^m[-g_i,g_i]. $$
\end{definition}

We remark that the interface dimension $K$ is a constant positive integer only related to the structure of the probability distribution. This fixed-dimensionality is the source of the approximation scheme.

For a compact convex region $C\subseteq\mathbb R^K$, the support function of $C$ is defined as
$$ \sigma_C(z)=\max_{y\in C}\inner{z}{y}. $$
It follows directly from the definition that the support function is monotone under set inclusion and positively homogeneous: for compact convex sets $C\subseteq C'\subseteq\mathbb R^K$ and any scalar $\lambda\ge0$,
$$ \sigma_C(z)\le\sigma_{C'}(z), \quad\text{and}\quad \sigma_{\lambda C}(z)=\lambda\,\sigma_C(z), \quad\text{for all }z\in\mathbb R^K. $$
We use this monotonicity to transfer zonotope containments to the support function in the direction $\gamma$.
For represented zonotopes, the support function has a particularly simple form, and TV-distance can be recovered by evaluating the support function at the direction $\gamma$.

%\begin{lemma}\label{lem:zonotope-support}
%Let $\mathcal G=\{g_1,\ldots,g_m\}\subseteq\mathbb R^K$ and $\+Z(\mathcal G)=\sum_{i=1}^m[-g_i,g_i]$.  Then for every $z\in\mathbb R^K$,
%$$ \sigma_{\+Z(\mathcal G)}(z) = \sum_{i=1}^m\abs{\inner{z}{g_i}}. $$
%\end{lemma}

%\begin{proof}
%Every point of $\+Z(\mathcal G)$ can be written as $\sum_i t_i g_i$ with $t_i\in[-1,1]$.  Therefore
%\begin{equation*} \sigma_{\+Z(\mathcal G)}(z) = \max_{t_i\in[-1,1]} \sum_i t_i\inner{z}{g_i} = \sum_i \abs{\inner{z}{g_i}}. \qedhere \end{equation*}
%\end{proof}

\begin{lemma}[Support-function representation of TV-distance]
\label{lem:tv-support}
For $\^P$ and $\^Q$ satisfying \Cref{def:tv-realization} with root generator multiset $\mathcal G^{\mathrm{root}}=\{g(x):x\in\Omega\}$ and vector $\gamma$, we have
$$ d_{\mathrm{TV}}(\^P,\^Q) = \frac{1}{2}\sigma_{\+Z(\mathcal G^{\mathrm{root}})}(\gamma). $$
\end{lemma}
\begin{proof}
  Every point of $\+Z(\mathcal G^{\mathrm{root}})$ can be written as $\sum_{x\in\Omega} t_x g(x)$ with $t_x\in[-1,1]$.  Therefore
  \begin{align}\label{eq:support-zonotope}
    \sigma_{\+Z(\mathcal G^{\mathrm{root}})}(\gamma)
    &= \max_{t\in[-1,1]^{\Omega}}\inner{\gamma}{\sum_{x\in\Omega}t_xg(x)}
    = \max_{t\in[-1,1]^{\Omega}}\sum_{x\in\Omega} t_x\inner{\gamma}{g(x)}
    = \sum_{x\in\Omega}\abs{\inner{\gamma}{g(x)}},
  \end{align}
where $t_x=1$ if $\inner{\gamma}{g(x)}\geq 0$ and $t_x=-1$ otherwise. By \Cref{def:tv-realization}, $\inner{\gamma}{g(x)}=\^P(x)-\^Q(x)$, thus 
\begin{equation*} 
  \sigma_{\+Z(\mathcal G^{\mathrm{root}})}(\gamma) = \sum_{x\in\Omega} \abs{\^P(x)-\^Q(x)} = 2d_{\mathrm{TV}}(\^P,\^Q). \qedhere 
\end{equation*}
\end{proof}

\subsection{Zonotope Compression}

The exact represented zonotopes can have exponentially many generators when $\abs{\Omega}$ is exponentially large. In this subsection, we give a way to compress a represented zonotope up to a relative factor. The compression algorithm follows a simple geometric intuition: we first put the zonotope in a position of bounded aspect ratio via an invertible linear transformation, then merge the generators with nearly the same direction, and finally map the compressed generators back by the inverse linear transformation.

\begin{lemma}\label{lem:lewis-rounding}
Let $\+G = \{g_1, \ldots, g_m\} \subseteq \^R^K$, and $\+Z(\+G)=\sum_{i=1}^m[-g_i,g_i]$ be the corresponding represented zonotope whose nonzero generators span an $r$-dimensional subspace $E$.  Write $B_2^r=\{x\in\mathbb R^r:\|x\|_2\le 1\}$ for the Euclidean unit ball in $\mathbb R^r$.  There is a deterministic algorithm that computes an invertible linear map $A:E\to\mathbb R^r$ in time $ O_K(m\log\log(m+2))$, such that
$$ B_2^r\subseteq A\+Z\subseteq2\sqrt r\,B_2^r.$$
\end{lemma}
The procedure for computing $A$ and its analysis are given in \Cref{sec:appendix}.

\begin{lemma}[Compression of represented zonotope]
\label{lem:compression}
Given $\mathcal{G} = \{g_1, \ldots, g_m\} \subseteq \^R^K$, let  
$\+Z=\+Z(\mathcal G)=\sum_{i=1}^m[-g_i,g_i]\subseteq\mathbb R^K$ be the corresponding represented zonotope.
Given $\eta\in(0,1)$, there exists an algorithm that constructs a compressed multiset $\underline{\+G} = \{\underline{g}_1, \ldots, \underline{g}_M\}$ with corresponding represented zonotope
$$\underline{\+Z}=\+Z(\underline{\+G})=\sum_{i=1}^M[-\underline g_i,\underline g_i] $$
such that
$\underline{\+Z}\subseteq \+Z\subseteq(1+\eta)\underline{\+Z}$
and $M=O(\eta^{-K})$, in time $O_K(m\log\log(m+2) + m\eta^{-K})$.
%$$ M\le\left\lceil\left(\frac{9L_{\le K}}{\eta}\right)^K\right\rceil,\qquad L_{\le K}:=\max_{1\le s\le K}2\kappa_s^{-1}\sqrt s. $$
\end{lemma}

\begin{proof}
We give the fixed-dimensional construction.  Let $E=\operatorname{span}\{g_1,\ldots,g_m\}$.  If $E=\{0\}$, return the zero zonotope.  Otherwise let $r=\dim E\le K$ and work inside $E$.

Apply \cref{lem:lewis-rounding} to compute $A:E\to\mathbb R^r$ such that
$B_2^r\subseteq A\+Z\subseteq2\sqrt r\,B_2^r$.
Write $y_i=Ag_i$.  The containment gives $\sigma_{A\+Z}(u)\in[1, 2\sqrt r]$ for every unit direction $u$. 

For algorithmic concreteness, let $\delta<1$ be a parameter to be fixed later. We use an explicit deterministic lattice net for unit directions.  Let
$$
L_\delta=\left\lceil\frac{\sqrt r}{\delta}\right\rceil+1,
\qquad
\+ N_\delta
=
\left\{
\frac{a}{\|a\|_2}:a\in\mathbb Z^r\setminus\{0\},\ \|a\|_\infty\le L_\delta+1
\right\}.
$$
This set is a $\delta$-net for unit directions.  Indeed, fix a unit direction $u\in\mathbb R^r$, and let $v_k$ be a nearest integer to $L_\delta u_k$ for each coordinate $k$.  Then $\|v/L_\delta-u\|_2\le \sqrt r/(2L_\delta)<\delta/2<1$, which implies $v\ne0$. Also $\|v\|_\infty\le L_\delta+1$, and hence $v/\|v\|_2\in\+ N_\delta$. Moreover,
\begin{align*}
\left\|\frac{v}{\|v\|_2}-u\right\|_2&=\left\|\frac{v/L_\delta}{\|v/L_\delta\|_2}-u\right\|_2 \leq\left\|\frac{v/L_\delta}{\|v/L_\delta\|_2}-\frac{v}{L_\delta}\right\|_2+\left\|\frac{v}{L_\delta}-u\right\|_2 \overset{\text{(i)}}{\leq}2\left\|\frac{v}{L_\delta}-u\right\|_2 <\delta,
\end{align*}
where (i) holds because $\frac{v/L_\delta}{\|v/L_\delta\|_2}$ and $u$ are unit vectors, while the first one has the same direction with $v$.
Thus every unit direction is within distance $\delta$ of some point in $\+ N_\delta$.  Its cardinality satisfies
$$
|\+ N_\delta|\le (2L_\delta+3)^r=O_r(\delta^{-r}).
$$

Enumerate $\+ N_\delta=\{p_1,\ldots,p_M\}$ with $M=|\+N_\delta|$.  Zero generators $y_i$ are discarded, since they do not affect the represented zonotope.  For each nonzero $y_i$, choose a nearest net point by brute force:
$$
j(i)=\arg\min_{j\in[M]}
\left\|
\frac{y_i}{\|y_i\|_2}-p_j
\right\|_2.
$$
By the net property, the chosen index satisfies $\left\|\frac{y_i}{\|y_i\|_2}-p_{j(i)}\right\|_2\le \delta$.
For each $j\in[M]$, aggregate all generators assigned to $p_j$ by setting
$$\underline{y}_j=\sum_{i:j(i)=j}y_i,
\qquad
\underline Y=\sum_{j=1}^M\left[-\underline{y}_j,\underline{y}_j\right].$$
If no generator is assigned to $p_j$, then $\underline y_j=0$.  Since $[-\underline{y}_j,\underline{y}_j]\subseteq\sum_{i:j(i)=j}[-y_i,y_i]$ for every $j\in[M]$, we have $\underline Y\subseteq A\+Z$.

Next we prove that $(1+\eta)\underline Y$ can cover $A\+Z$ by choosing a proper $\delta$.  Fix a unit direction $u$.  For an index $j\in[M]$, write $y_i=\|y_i\|q_i$ for all $i$ with $j(i)=j$, where $\|q_i-p_j\|\le\delta$.  Set
$$ \alpha_j=\inner{u}{p_j}, \qquad h_i=\inner{u}{q_i-p_j}. $$
Then $\abs{h_i}\le\delta$. By \eqref{eq:support-zonotope}, the contribution of the original generators in this group to the support function is
$$\sigma_{\sum_{i:j(i)=j}[-y_i, y_i]}(u) = \sum_{i:j(i)=j}\left|\langle u, y_i\rangle\right| = \sum_{i:j(i)=j}\|y_i\|\abs{\alpha_j+h_i}, $$
whereas the contribution after aggregation is
$$\sigma_{[-\underline{y}_j, \underline{y}_j]}(u) = \left| \langle u, \underline{y}_j \rangle \right| = \left|\sum_{i:j(i)=j}\|y_i\|(\alpha_j+h_i)\right|. $$
If $\abs{\alpha_j}>\delta$, all terms $\alpha_j+h_i$ have the same sign and no loss occurs. If $\abs{\alpha_j}\leq\delta$, the loss is at most $2\delta\sum_{i:j(i)=j}\|y_i\|$. Thus
\begin{align*}
  \sigma_{A\+Z}(u)-\sigma_{\underline Y}(u) &=\sum_{j=1}^M\left(\sum_{i:j(i)=j}\|y_i\||\alpha_j+h_i| - \left|\sum_{i:j(i)=j}\|y_i\|(\alpha_j+h_i)\right|\right)\\
  &\leq \sum_{j=1}^M2\delta\sum_{i:j(i)=j}\|y_i\| = 2\delta \sum_{i=1}^{m}\|y_i\|.
\end{align*} 

Let $c_r=\E[v]{|\inner{v}{e}|}$ be a constant related only to $r$, where $v$ is a uniformly random unit vector in $\mathbb R^r$, and $e$ is an arbitrary fixed unit vector in $\^R^r$\footnote{In fact, by a standard spherical-coordinate calculation, $c_r=\frac{\Gamma(r/2)}{\sqrt\pi \Gamma((r+1)/2)}$.}.  By rotational invariance,
$$
\E[v]{\sigma_{A\+Z}(v)}
=\sum_{i=1}^m \E[v]{|\inner{v}{y_i}|}
=c_r\sum_{i=1}^m\|y_i\|.
$$
Since $\sigma_{A\+Z}(v)\leq2\sqrt r$ for every unit vector $v$, this gives
$$
\sum_{i=1}^m\|y_i\|\leq2\frac{\sqrt r}{c_r}.
$$
Let
$$ \delta=\frac{c_r\eta}{8\sqrt r},$$
then the preceding loss bound implies
$$
\sigma_{\underline Y}(u)\ge\sigma_{A\+Z}(u)-4\delta \frac{\sqrt r}{c_r}
=\sigma_{A\+Z}(u)-\frac{\eta}{2}
\geq (1-\eta/2)\sigma_{A\+Z}(u),
$$
where the last inequality uses $\sigma_{A\+Z}(u)\ge1$.  Since $(1-\eta/2)(1+\eta)\ge1$ for $0<\eta<1$, this implies $A\+Z\subseteq(1+\eta)\underline Y$. Mapping back by $A^{-1}$, define $\underline g_j=A^{-1}\underline y_j$ for $j\in[M]$ and return $\underline{\+G}=\{\underline g_j:j\in[M]\}$. Its represented zonotope is $\underline{\+Z}=A^{-1}\underline Y$, so it satisfies the desired containments.

The number of returned generators is $M=|\mathcal N_\delta|=O_r(\delta^{-r})$.  With $\delta=c_r\eta/(8\sqrt r)$ and $r\le K$, this gives $M=O_K(\eta^{-K})$. By \Cref{lem:lewis-rounding}, computing $A$ requires $O_K(m\log\log(m+2))$ time.  The lattice net has size $O_K(\eta^{-K})$, so exhaustive search for each $j(i)$ costs $O_K(m\eta^{-K})$ time in total; forming the aggregates and mapping them back costs no more. Thus the total running time is $O_K(m\log\log(m+2) + m\eta^{-K})$.
\end{proof}

\subsection{Admissible Gate}
Since the root generator multiset $\+G^{\mathrm{root}}$ may contain an exponential number of vectors, we cannot form it explicitly. Instead, we build this multiset recursively from small leaf multisets. We therefore need a condition ensuring that approximate inputs may be passed through one step of the recursion without destroying the zonotope containment.

\begin{definition}[Admissible gate]
\label{def:admissible-gate}
An \emph{arity-$a$ gate} $G$ maps finite generator multisets $\+G_1,\ldots,\+G_a\subseteq\mathbb R^K$ to a finite generator multiset $G(\+G_1,\ldots,\+G_a)\subseteq\mathbb R^K$; given explicit lists of the inputs, the output can be listed in time $T_G(\mu)$, where $\mu$ upper bounds the total size of the input and output multisets.

We say $G$ is \emph{admissible} with \emph{error map} $\phi:\mathbb R_+^a\to\mathbb R_+$ if it satisfies the following stability property: for every $\eta\in(0,1)$, whenever exact input multisets $\+G_i$ and approximate multisets $\underline{\+G}_i$ satisfy
$$ \+Z(\underline{\+G}_i)\subseteq \+Z(\+G_i)\subseteq(1+\eta)^{s_i}\+Z(\underline{\+G}_i) \qquad\text{for each }i\in[a] $$
with some $s_i\in\mathbb R_+$, then the exact output $\+Z=\+Z(G(\+G_1,\ldots,\+G_a))$ and the approximate output $\underline{\+Z}=\+Z(G(\underline{\+G}_1,\ldots,\underline{\+G}_a))$ satisfy
$$ \underline{\+Z}\subseteq \+Z\subseteq(1+\eta)^{\phi(s_1,\ldots,s_a)}\underline{\+Z}. $$

%An arity-$a$ gate $G$ takes finite generator multisets $\mathcal G_1,\ldots,\mathcal G_a\subseteq\mathbb R^K$ as input, and outputs another finite generator multiset $G(\mathcal G_1,\ldots,\mathcal G_a)\subseteq\mathbb R^K$ in time $T_G$.

%We call a gate $G$ admissible with error map $\phi:\mathbb R_+^a\to\mathbb R_+$ if the following stability property holds.
%For $\eta\in(0,1)$, every exact input multiset $\mathcal G_i$, with corresponding approximate multiset $\underline{\mathcal G}_i$ satisfies
%$$ \+Z(\underline{\mathcal G}_i)\subseteq \+Z(\mathcal G_i)\subseteq(1+\eta)^{s_i}\+Z(\underline{\mathcal G}_i), $$
%for some $s_i\in\^R_+$, the exact output of applying $G$ to the exact input multisets
%$\+Z=\+Z(G(\mathcal G_1,\ldots,\mathcal G_a))$, and the output of applying $G$ to the approximate multisets $\underline{\+Z}=\+Z(G(\underline{\mathcal G}_1,\ldots,\underline{\mathcal G}_a))$, satisfy
%$$ \underline{\+Z}\subseteq \+Z\subseteq(1+\eta)^{\phi(s_1,\ldots,s_a)}\underline{\+Z}. $$
\end{definition}

\subsection{Abstract FPTAS}

We assume that the error maps are explicitly computable within the corresponding gate-evaluation bounds $T_{G_i}(\mu_i)$, and the framework gives an FPTAS whenever the bound in \Cref{thm:abstract-fptas} is polynomial in the input size and $1/\varepsilon$ for fixed $K$.

\begin{assumption}
\label{assumption:admissible-gate-circuit}
Given $\^P$ and $\^Q$, let $\+G^{\mathrm{root}}$ be a root generator multiset with corresponding vector $\gamma$ as defined in \Cref{def:tv-realization}. There is a finite collection of explicit leaf generator multisets
$$ \+L=\{\+G^{\mathrm{leaf}}_1,\ldots,\+G^{\mathrm{leaf}}_\ell\} $$
such that $\gamma$ can be computed and $\+L$ can be constructed in time $T_0$, and a sequence of $N$ admissible gate applications $G_1,\ldots,G_N$ with error maps $\phi_1,\ldots,\phi_N$.  At step $i$, apply $G_i$ to obtain $\+G_i$ on input multisets $\+G^{(i)}_1,\ldots,\+G^{(i)}_{a_i}$, each of which is either in $\+L$ or $\+G_j$ for some $j<i$. After applying $G_N$, the final output is $\+G_N=\mathcal G^{\mathrm{root}}$.
\end{assumption}

For a gate sequence satisfying \cref{assumption:admissible-gate-circuit}, define its error weights recursively as follows.  Each leaf multiset $\+G^{\mathrm{leaf}}_j$ has weight $0$.  If $\+G^{(i)}_j$ has weight $s^{(i)}_{j}$ for $j\in[a_i]$, then the output multiset $\+G_i$ has weight
\begin{align}\label{eq:error-weight-recursion}
  s_i=\phi_i\left(s^{(i)}_1,\ldots,s^{(i)}_{a_i}\right)+1. 
\end{align}
Let $S=s_N$ be the final error weight of $\+G_N = \+G^{\mathrm{root}}$.

\begin{algorithm}[ht]
  \caption{Abstract FPTAS for $\^P$ and $\^Q$ satisfying \Cref{assumption:admissible-gate-circuit}}
  \label{alg:abstract-FPTAS}
  \SetKwInOut{Input}{Input}
  \SetKwInOut{Output}{Output}
  \Input{Description of $\^P$ and $\^Q$, $\gamma\in \^R^K$ as in \Cref{def:tv-realization}, $\epsilon \in (0,1)$, and the gate sequence from \Cref{assumption:admissible-gate-circuit}}
  \Output{An approximation $\hat{d}$ of $\DTV{\^P}{\^Q}$ such that $\hat{d} \le d_{\mathrm{TV}}(\^P,\^Q) \le (1+\epsilon)\hat{d}$}
  Construct the leaf family $\+L=\{\+G^{\mathrm{leaf}}_1,\ldots,\+G^{\mathrm{leaf}}_\ell\}$, let $\underline{\+G}^{\mathrm{leaf}}_j = \+G^{\mathrm{leaf}}_j$ for all $j\in[\ell]$\;
  Recursively compute the error weights $s_1,\ldots,s_N$ for the gate sequence using \eqref{eq:error-weight-recursion}\;
  Set $S \gets s_N$, and $\eta\gets (1+\epsilon)^{1/S}-1$\;
  \For{$i=1,\ldots,N$}{
    Compute $\widetilde{\+G}_i \gets G_i(\underline{\+G}^{(i)}_1,\ldots,\underline{\+G}^{(i)}_{a_i})$\label{line:gate-application}\;
    Compress $\widetilde{\+G}_i$ using \Cref{lem:compression} with tolerance $\eta$, and store the resulting multiset as $\underline{\+G}_i$\;
  }
  \Return{$\hat{d}=\frac12\sigma_{\+Z(\underline{\+G}_N)}(\gamma)$}\;
\end{algorithm}

\begin{theorem}
\label{thm:abstract-fptas}
Let $\^P$ and $\^Q$ be probability distributions satisfying \Cref{assumption:admissible-gate-circuit}, vector $\gamma \in \^R^K$ as in \Cref{def:tv-realization}, and $\epsilon \in (0,1)$. Then Algorithm \ref{alg:abstract-FPTAS} outputs an approximation $\hat{d}$ of $\DTV{\^P}{\^Q}$ such that
$$ \hat{d} \le d_{\mathrm{TV}}(\^P,\^Q) \le (1+\epsilon)\hat{d}. $$
Moreover, if $M_{G_i} \defeq |\widetilde{\+G}_i|$ denotes the output size of gate $G_i$ in Line \ref{line:gate-application} of Algorithm \ref{alg:abstract-FPTAS}, and $T_{G_i}(\mu_i)$ is the corresponding gate-evaluation time on inputs and output of total size $\mu_i$, then the total running time is
$$ O_K\left(T_0+\sum_{i=1}^N\left(T_{G_i}(\mu_i)+M_{G_i}\log\log(M_{G_i}+2)+M_{G_i}(S/\epsilon)^K\right)\right). $$
\end{theorem}

\begin{proof}
We prove by induction on $i$ that after the $i$-th iteration of Algorithm \ref{alg:abstract-FPTAS},
\begin{align}\label{eq:invariant}
  \+Z(\underline{\+G}_i)\subseteq \+Z(\+G_i)\subseteq(1+\eta)^{s_i}\+Z(\underline{\+G}_i),
\end{align}
where $\+G_i$ is the exact output of $G_i$ in \Cref{assumption:admissible-gate-circuit}.

For every leaf multiset, the algorithm stores it exactly, so \eqref{eq:invariant} holds with weight $0$.  Suppose \eqref{eq:invariant} holds for all $h<i$.  In step $i$, each input $\underline{\+G}^{(i)}_j$ is either a leaf, in which case the stored approximation is exact, or an earlier stored approximation output $\underline{\+G}_h$ in step $h$. Thus the stored inputs satisfy the hypotheses of admissibility with weights $s^{(i)}_1,\ldots,s^{(i)}_{a_i}$.  Therefore by \Cref{def:admissible-gate}, the pre-compression output $\widetilde{\+G}_i$ satisfies
$$ \+Z(\widetilde{\+G}_i)\subseteq \+Z(\+G_i)\subseteq(1+\eta)^{\phi_i(s^{(i)}_1,\ldots,s^{(i)}_{a_i})}\+Z(\widetilde{\+G}_i). $$
The compression step gives
$$ \+Z(\underline{\+G}_i)\subseteq \+Z(\widetilde{\+G}_i)\subseteq(1+\eta)\+Z(\underline{\+G}_i). $$
Combining the two containments and using \eqref{eq:error-weight-recursion} proves the induction claim.

At $i=N$, $\+G_N=\+G^{\mathrm{root}}$.  Hence
$$ \+Z(\underline{\+G}_N)\subseteq \+Z(\+G^{\mathrm{root}})\subseteq(1+\eta)^S \+Z(\underline{\+G}_N)=(1+\epsilon)\+Z(\underline{\+G}_N). $$
Since support functions are monotone under set inclusion, the containment above implies
$$
\sigma_{\+Z(\underline{\+G}_N)}(\gamma)
\le \sigma_{\+Z(\+G^{\mathrm{root}})}(\gamma)
\le (1+\epsilon)\sigma_{\+Z(\underline{\+G}_N)}(\gamma).
$$
Together with \Cref{lem:tv-support}, this gives
$$ \hat{d}=\frac12\sigma_{\+Z(\underline{\+G}_N)}(\gamma)\le d_{\mathrm{TV}}(\^P,\^Q)\le(1+\epsilon)\hat{d}. $$

It remains to bound the running time.  Constructing $\+L$ and computing $\gamma$ costs $T_0$.  At step $i$, Line \ref{line:gate-application} costs $T_{G_i}(\mu_i)$ and produces $M_{G_i}$ generators, where $\mu_i$ is the total size of the inputs and output of $G_i$.  By \Cref{lem:compression}, compressing this multiset with tolerance $\eta$ costs
$$ O_K(M_{G_i}\log\log(M_{G_i}+2)+M_{G_i}\eta^{-K})
=O_K(M_{G_i}\log\log(M_{G_i}+2)+M_{G_i}(S/\epsilon)^K). $$
Summing over $i=1,\ldots,N$ gives the stated running-time bound.
\end{proof}

\section{Applications}

The preceding section gives an FPTAS for any pair of distributions whose TV-realization admits a short admissible-gate sequence with efficient compression.  We now instantiate this abstract framework for several concrete distribution classes by specifying the interface vector, the leaf generator multisets, and the required admissible gates.

\subsection{Mixture of Product Distributions}\label{sec:mix-of-prod}
Let $[q]=\{1,\ldots,q\}$ and $\Omega=[q]^n$. For constants $k_1,k_2\geq1$, let $\^P$ and $\^Q$ be mixtures of product distributions on $\Omega$ with $k_1$ and $k_2$ components, where $$\^P = \sum_{s=1}^{k_1} \alpha^{(s)} \bigotimes_{i=1}^n \^P^{(s)}_i \text{ and } \^Q = \sum_{t=1}^{k_2} \beta^{(t)} \bigotimes_{i=1}^n \^Q^{(t)}_i.$$ 
Set $K = k_1+k_2$, $\gamma = (\alpha^{(1)}, \ldots, \alpha^{(k_1)}, -\beta^{(1)}, \ldots, -\beta^{(k_2)})$, and $\+G^{\mathrm{root}} = \{g(x):x\in\Omega\}$, where
$$g(x) = \left(\^P^{(1)}(x), \ldots, \^P^{(k_1)}(x), \^Q^{(1)}(x), \ldots, \^Q^{(k_2)}(x)\right)\quad \forall x\in\Omega.$$

We now substitute mixtures of product distributions into the abstract framework. To verify \Cref{assumption:admissible-gate-circuit}, we first introduce the diagonal-sum gate that will be used to generate the root multiset.

\begin{definition}[Diagonal-sum gate]\label{def:diag-sum-gate}
  Given a set $D$ of diagonal linear maps $D_1,\ldots,D_q:\^R^K\to\^R^K$, define the arity-$q$ gate $G_D$ as follows: it applies $D_v$ to every generator in the $v$-th input multiset, and then takes the multiset union over all $v\in[q]$:
  $$ G_D(\mathcal G_1,\ldots,\mathcal G_q)=\biguplus_{v\in[q]}D_v\+G_v=\biguplus_{v\in[q]}\left\{D_vg:g\in\mathcal G_v\right\}. $$
\end{definition}

For product distributions, revealing one coordinate simply multiplies each
component probability by the corresponding one-coordinate marginal.  Thus each
possible symbol $v\in[q]$ applies a diagonal scaling to the current generator
multiset, and the next prefix multiset is the union over all choices of $v$.
The diagonal-sum gate abstracts exactly this operation.
  
\begin{lemma}\label{lem:admissible-diag-sum}
  The gate $G_D$ is admissible with error map $\phi_D(s_1,\ldots,s_q)=\max_{v\in[q]}s_v$. Moreover, its evaluation time is $T_{G_D}(\mu)=O_K(\mu)$, where $\mu$ is the total size of the input and output multisets; in particular, if every input $\+G_v$ has size at most $M$, then $T_{G_D}(2qM)=O_K(qM)$.
\end{lemma}

\begin{proof}
  By definition of $G_D$, the represented zonotope of the output is
  $$ \+Z\left(G_D(\mathcal G_1,\ldots,\mathcal G_q)\right)=\sum_{v\in[q]}\+Z(D_v\mathcal G_v) = \sum_{v\in[q]}D_v\+Z(\mathcal G_v), $$
  here we use Minkowski sum for summation between zonotopes.

  For every $v \in [q]$, suppose there exists $\underline{\+G}_v$ such that
  $$ \+Z(\underline{\mathcal G}_v)\subseteq \+Z(\mathcal G_v)\subseteq(1+\eta)^{s_v}\+Z(\underline{\mathcal G}_v),$$
  then $$\+Z(G_D(\underline{\mathcal G}_1,\ldots,\underline{\mathcal G}_q)) = \sum_{v\in[q]}D_v\+Z(\underline{\+G}_v) \subseteq \sum_{v\in[q]}D_v\+Z(\+G_v) = \+Z\left(G_D(\mathcal G_1,\ldots,\mathcal G_q)\right).$$
  For the other direction, since for any $v\in[q]$, $\phi_D(s_1,\ldots,s_q)\geq s_v$, thus $$\+Z(\mathcal G_v)\subseteq (1+\eta)^{s_v}\+Z(\underline{\mathcal G}_v)\subseteq(1+\eta)^{\phi_D(s_1,\ldots,s_q)}\+Z(\underline{\mathcal G}_v),$$
  which implies
  $$\sum_{v\in[q]}D_v\+Z(\+G_v) \subseteq (1+\eta)^{\phi_D(s_1,\ldots,s_q)}\sum_{v\in[q]}D_v\+Z(\underline{\+G}_v).$$

  Each of the $\sum_{v\in[q]}|\+G_v|$ input generators produces one output generator through a single multiplication $D_vg$, which takes $O_K(1)$ time for constant $K$.  Hence the total size satisfies $\mu=\Theta\big(\sum_{v\in[q]}|\+G_v|\big)$ and $T_{G_D}(\mu)=O_K(\mu)$; in particular, if every input has size at most $M$, then $\mu=O(qM)$ and $T_{G_D}(\mu)=O_K(qM)$.
\end{proof}

\begin{proof}[Proof of~\Cref{thm:mix-of-prod}]
  We verify \Cref{assumption:admissible-gate-circuit} and record the parameters in \Cref{thm:abstract-fptas}.  For $a\in[K]$ and $i\in[n]$, define the merged one-coordinate marginals by
  $$
  R_{a,i}=
  \begin{cases}
    \^P_i^{(a)}, & a\le k_1,\\
    \^Q_i^{(a-k_1)}, & a>k_1.
  \end{cases}
  $$
  Thus $R_{a,i}$ is the marginal distribution on the $i$-th coordinate.  For each $i\in[n]$, let $D_i=(D_{i,1},\ldots,D_{i,q})$, where
  $$
  D_{i,v}=\operatorname{diag}\left(R_{1,i}(v),\ldots,R_{K,i}(v)\right).
  $$

  The leaf family consists of one multiset
  $$
  \+L=\{\+H_0\},\quad \text{where } \+H_0=\{(1,1,\ldots,1)\}\subseteq\mathbb R^K.
  $$
  Hence $T_0=O_K(1)$.  For $i=1,\ldots,n$, apply the diagonal-sum gate $G_{D_i}$ to $q$ copies of the previously constructed multiset:
  $$
  \+H_i=G_{D_i}(\+H_{i-1},\ldots,\+H_{i-1}).
  $$
  This gate sequence has length
  $N=n$,
  and its final output is $\+G_N=\+H_n$.

  By induction on $i$, $\+H_i$ is exactly the generator multiset
  $$
  \left\{\left(\prod_{j=1}^i R_{a,j}(x_j)\right)_{a=1}^K:(x_1,\ldots,x_i)\in[q]^i\right\}.
  $$
  The base case is the empty prefix $\+H_0=\{(1,\ldots,1)\}$.  For the induction step, applying $G_{D_i}$ chooses the next symbol $v=x_i$ and multiplies the $a$-th coordinate by $R_{a,i}(v)$, so a generator corresponding to $x_1,\ldots,x_{i-1}$ is sent to one corresponding to $x_1,\ldots,x_i$:
  $$
  R_{a,i}(v)\prod_{j=1}^{i-1} R_{a,j}(x_j)
  =
  \prod_{j=1}^i R_{a,j}(x_j),
  $$
  and every prefix is obtained in this way.
  Therefore the final root output is
  $$
  \+H_n=\left\{\left(\prod_{j=1}^n R_{a,j}(x_j)\right)_{a=1}^K:x\in[q]^n\right\}=\+G^{\mathrm{root}}.
  $$
  Thus the final gate sequence realizes the root generator multiset in \Cref{def:tv-realization} with vector $\gamma$.

  By the preceding lemma, every gate in this sequence is admissible with error map
  $$\phi_D(s_1,\ldots,s_q)=\max_{v\in[q]}s_v.$$
  The recursive procedure of computing error weights gives $s_i = i$, and thus $S=s_n=n$.

  It remains to substitute the parameters into the running-time bound of \Cref{thm:abstract-fptas}.  The tolerance used by Algorithm \ref{alg:abstract-FPTAS} is
  $$
  \eta=(1+\varepsilon)^{1/S}-1=(1+\varepsilon)^{1/n}-1,
  $$
  so $\eta^{-1}=O(n/\varepsilon)$.  By \Cref{lem:compression}, each stored compressed multiset has size
  $$
  M=O_K(\eta^{-K})=O_K((n/\varepsilon)^K).
  $$
  Each diagonal-sum gate has $q$ inputs of size at most $M$, so for every gate in the sequence,
  $$
  M_{G_D}\le qM, \qquad \mu_i = O(qM), \qquad T_{G_D}(\mu_i)=O_K(qM).
  $$
  Therefore \Cref{thm:abstract-fptas} gives total running time
  \begin{equation*}
    O_K\left(1+n\left(qM+qM\log\log(qM)+qM(n/\varepsilon)^K\right)\right)
    =\widetilde O_K\left(nq(n/\varepsilon)^{2K}\right).\qedhere
  \end{equation*}
\end{proof}

\subsection{Mixture of Markov Chains}\label{sec:mix-of-mc}

Let $\Omega=[q]^n$.  Fix constants $k_1,k_2\geq1$.  For each $s\in[k_1]$, let $\^P^{(s)}$ be a Markov-chain distribution on $\Omega$ with initial distribution $\^P_1^{(s)}$ over $[q]$, and transition matrices $\^P^{(s)}_{2|1},\ldots,\^P^{(s)}_{n|n-1}$, so that for any $x\in[q]^n$,
$$\^P(x)=\sum_{s=1}^{k_1}\alpha^{(s)}\^P^{(s)}(x) =\sum_{s=1}^{k_1}\alpha^{(s)} \^P_1^{(s)}(x_1)\prod_{i=2}^n \^P^{(s)}_{i|i-1}(x_i|x_{i-1}). $$
Similarly, 
$$\^Q(x)=\sum_{t=1}^{k_2}\beta^{(t)}\^Q^{(t)}(x) =\sum_{t=1}^{k_2}\beta^{(t)} \^Q_1^{(t)}(x_1)\prod_{i=2}^n \^Q^{(t)}_{i|i-1}(x_i|x_{i-1}).$$

Write $\gamma = (\alpha^{(1)}, \ldots, \alpha^{(k_1)}, -\beta^{(1)}, \ldots, -\beta^{(k_2)})$, $\+G^{\mathrm{root}} = \{g(x): x\in\Omega\}$ where 
$$g(x) = \left(\^P^{(1)}(x), \ldots, \^P^{(k_1)}(x), \^Q^{(1)}(x), \ldots, \^Q^{(k_2)}(x)\right), \quad \forall x \in \Omega.$$
Then $\^P$ and $\^Q$ satisfy \Cref{def:tv-realization} with vector $\gamma$ and multiset $\+G^{\mathrm{root}}$. 

We now extend the preceding product-distribution construction to mixtures of Markov chains. We still use the diagonal-sum gate as defined in \Cref{def:diag-sum-gate}. In the product case, one multiset suffices at each level because the coordinate marginals are independent of the past. For Markov chains, the transition at time $i$ depends on the previous state $x_{i-1}$. We therefore build the generators backward, from suffixes to the full sequence: at each level $i$, we keep $q$ multisets indexed by the possible value $u\in[q]$ of $x_{i-1}$. The multiset indexed by $u$ represents suffix generators conditioned on $x_{i-1}=u$.

\begin{proof}[Proof of \Cref{thm:mix-of-mc}]
  For $a\in[K]$, define the merged kernels by
  $$
  R_{a,1}=
  \begin{cases}
    \^P_1^{(a)}, & a\le k_1,\\
    \^Q_1^{(a-k_1)}, & a>k_1,
  \end{cases}
  \qquad
  R_{a,i}=
  \begin{cases}
    \^P^{(a)}_{i|i-1}, & a\le k_1,\\
    \^Q^{(a-k_1)}_{i|i-1}, & a>k_1,
  \end{cases}
  \text{ for } i=2,\ldots,n.
  $$
  Thus $R_{a,1}$ is a distribution on the first coordinate, while $R_{a,i}$ is a transition kernel for $i\ge2$.  For $i=2,\ldots,n$ and $u\in[q]$, let $D_{i,u}=(D_{i,u,1},\ldots,D_{i,u,q})$, where
  $$
  D_{i,u,v}=\operatorname{diag}\left(R_{1,i}(v|u),\ldots,R_{K,i}(v|u)\right),
  $$
  and let $D_1=(D_{1,1},\ldots,D_{1,q})$, where
  $$
  D_{1,v}=\operatorname{diag}\left(R_{1,1}(v),\ldots,R_{K,1}(v)\right).
  $$
  The matrix $D_{i,u,v}$ attaches one transition from boundary state $u$ to the next state $v$ simultaneously in all mixture coordinates. Thus the same diagonal-sum gate suffices; the only change from the product case is that the current multiset is indexed by the boundary state, and there are $q$ of them in each layer.

  The leaf family is
  $$
  \+L=\{\+H_{n+1}^u:u\in[q]\},\quad \text{where } \+H_{n+1}^u=\{(1,1,\ldots,1)\}\subseteq\mathbb R^K.
  $$
  Hence the initialization time $T_0=O_K(q)$.  For $i=n,n-1,\ldots,2$ and $u\in[q]$, apply the diagonal-sum gate $G_{D_{i,u}}$ to the input multisets indexed by the next state.  More explicitly, if $\+H_{i+1}^v$ has already been constructed for every $v\in[q]$, set
  $$
  \+H_i^u=G_{D_{i,u}}(\+H_{i+1}^1,\ldots,\+H_{i+1}^q).
  $$
  Finally, apply one more diagonal-sum gate $G_{D_1}$ to $\+H_2^1,\ldots,\+H_2^q$ and call the output $\+H_1$.  This gate sequence has length
  $$N=q(n-1)+1.$$

  By backward induction on $i$, we show that $\+H_i^u$ is exactly the generator multiset
  $$
  \left\{\left(\prod_{j=i}^n R_{a,j}(x_j|x_{j-1})\right)_{a=1}^K:(x_i,\ldots,x_n)\in[q]^{n-i+1}, x_{i-1}=u\right\}.
  $$
  That is, condition on $x_{i-1} = u$, each generator records the componentwise conditional probability of one possible suffix $x_i,\ldots,x_n$; its $a$-th coordinate is the probability of generating that suffix under component $a$.
  
  For $i=n+1$, this is the empty suffix, represented by $\+H_{n+1}^u=\{(1,\ldots,1)\}$.  Suppose the claim holds for $\+H_{i+1}^v$ for every $v\in[q]$.  Then $\+H_i^u$ is the multiset union of $D_{i,u,v}g$ over $v\in[q]$ and $g\in\+H_{i+1}^v$.  If $g$ corresponds to $x_{i+1},\ldots,x_n$ with boundary $x_i=v$, then its new $a$-th coordinate is
  $$
  R_{a,i}(v|u)\prod_{j=i+1}^n R_{a,j}(x_j|x_{j-1})
  =
  \prod_{j=i}^n R_{a,j}(x_j|x_{j-1}),
  $$
  after setting $x_i=v$.  Conversely, every suffix with $x_{i-1}=u$ is obtained by choosing $v=x_i$ and then the tail $x_{i+1},\ldots,x_n$.  Thus the description holds at level $i$.
  Therefore the final root output is
  $$
  \+H_1=\left\{\left(R_{a,1}(x_1)\prod_{j=2}^n R_{a,j}(x_j|x_{j-1})\right)_{a=1}^K:x\in[q]^n\right\}=\+G^{\mathrm{root}}.
  $$
  Thus the final gate sequence realizes the root generator multiset in \Cref{def:tv-realization} with vector $\gamma$.

  The rest of the parameter calculation is the same as in the product case. By \Cref{lem:admissible-diag-sum}, all gates are admissible with error map $\phi_D(s_1,\ldots,s_q)=\max_{v\in[q]}s_v$, and the final error weight is $S=n$.  Thus \Cref{thm:abstract-fptas} gives
  $$
  \hat{d}\leq \DTV{\^P}{\^Q}\leq(1+\varepsilon)\hat{d}.
  $$
  With $\eta=(1+\varepsilon)^{1/n}-1$, each stored compressed multiset has size
  $$
  M=O_K((n/\varepsilon)^K).
  $$
  Since there are $N=q(n-1)+1=O(nq)$ diagonal-sum gates, and each gate has total input-output size $O(qM)$, the evaluation time is $O_K(qM)$, substituting these parameters into \Cref{thm:abstract-fptas} gives the running time
  \begin{equation*}
    \widetilde O_K\left(nq^2(n/\varepsilon)^{2K}\right).\qedhere
  \end{equation*}
\end{proof}

\subsection{Latent-Tree Ising Models}\label{sec:latent-tree}

Let $T=(V,E)$ be an undirected tree with leaf set $L\subseteq V$.  We consider two Ising models $(T,J^{\^P},h^{\^P})$ and $(T,J^{\^Q},h^{\^Q})$, specifying two Gibbs distributions $\^P$ and $\^Q$ over $\{\pm1\}^V$.  For $(T,J^{\^P},h^{\^P})$, the full spin distribution is
$$\^P(X_V=x_V) = \frac{1}{Z_{\^P}} \exp\left( \sum_{\{u,v\}\in E}J^{\^P}_{uv}x_u x_v + \sum_{u\in V}h^{\^P}_u x_u \right), \quad x_V\in\{\pm1\}^V.$$
Let $\^P_L$ be the marginal distribution of $\^P$ over $L$, i.e.,
$$\^P_L(x_L) = \sum_{\substack{y\in\{\pm1\}^V:\\ y_L=x_L}} \^P(X_V=y), \quad \forall x_L \in \{\pm1\}^L.$$
$\^Q$ and $\^Q_L$ are defined analogously. Our goal is to approximate $\DTV{\^P_L}{\^Q_L}$.

The main difficulty is that the leaf marginal is obtained after summing over
all hidden spins.  We avoid normalizing at every subtree.  Instead, for each
subtree we keep the two unnormalized contributions corresponding to the spin
of its root.  The global normalization will be inserted only once, at the root.

For $|V|\le2$, compute the TV-distance exactly by enumeration; otherwise, fix an arbitrary internal root $r$ of $T$ and orient all edges away from $r$. Denote the subtree rooted at $u$ by $T_u = (V_u, E_u)$, and let $L_u = L \cap V_u$. For any $x_{L_u}\in\{\pm1\}^{L_u}$, define the restricted subtree partition function
\begin{equation}\label{eq:conditional-partition-function}
  Z^u_R(s,x_{L_u})=\sum_{\substack{y\in\{\pm1\}^{V_u}:\\y_u=s,\ y_{L_u}=x_{L_u}}}\exp\left(\sum_{\{i,j\}\in E_u}J^R_{ij}y_i y_j+\sum_{i\in V_u}h^R_i y_i
\right),\forall s\in\{\pm1\},R\in\{\^P,\^Q\}.\end{equation}

This suggests a four-dimensional interface.  For each boundary spin
$s\in\{\pm1\}$ we keep the unnormalized contribution under both models
$\^P$ and $\^Q$.  Thus the coordinates record exactly the information needed
to later compare the two leaf marginals.
Set $K = 4$, and index the four coordinates by pairs $(R,s)\in\{\^P,\^Q\}\times\{\pm1\}$, with coordinate order $(\^P,+1),(\^P,-1),(\^Q,+1),(\^Q,-1)$, define
\begin{align}\label{eq:latent-g}
  g_u(x_{L_u})
  =
  \left(
  Z^u_{\^P}(+1,x_{L_u}),
  Z^u_{\^P}(-1,x_{L_u}),
  Z^u_{\^Q}(+1,x_{L_u}),
  Z^u_{\^Q}(-1,x_{L_u})
  \right)\in\^R^4.
\end{align}

Let $Z_{\^P}$ and $Z_{\^Q}$ be the full partition functions of the two Ising models, and set
$$
\gamma=\left(\frac1{Z_{\^P}},\frac1{Z_{\^P}},-\frac1{Z_{\^Q}},-\frac1{Z_{\^Q}}\right).
$$
Let $\+G^{\mathrm{root}}=\{g_r(x_L):x_L\in\{\pm1\}^L\}$, then
$$
\inner{\gamma}{g_r(x_L)}
=
\frac{Z^r_{\^P}(+1,x_L)+Z^r_{\^P}(-1,x_L)}{Z_{\^P}}
-
\frac{Z^r_{\^Q}(+1,x_L)+Z^r_{\^Q}(-1,x_L)}{Z_{\^Q}}
=\^P_L(x_L)-\^Q_L(x_L).
$$
Thus $\+G^{\mathrm{root}}$ satisfies \Cref{def:tv-realization} with vector $\gamma$.

The gate sequence describes the bottom-up partition-function dynamic program. The initial leaf family is
$\+L=\{\+G_u^{(0)}:u\in V\}$ defined as follows:
For an unobserved internal vertex $u\notin L$, set
$$
\+G_u^{(0)}
=
\left\{
\left(
e^{h_u^{\^P}},
e^{-h_u^{\^P}},
e^{h_u^{\^Q}},
e^{-h_u^{\^Q}}
\right)
\right\}.
$$
For a leaf $\ell\in L$, set
\begin{align}\label{eq:G-leaf}
\+G_\ell^{(0)}
=
\left\{
\left(
e^{h_\ell^{\^P}}\mathbf 1_{\{x=+1\}},
e^{-h_\ell^{\^P}}\mathbf 1_{\{x=-1\}},
e^{h_\ell^{\^Q}}\mathbf 1_{\{x=+1\}},
e^{-h_\ell^{\^Q}}\mathbf 1_{\{x=-1\}}
\right)
:x\in\{\pm1\}
\right\}.
\end{align}
Thus $\+G_u^{(0)}$ contains the local vertex field, and the edge interactions will be inserted by the gates.

\begin{definition}[Edge-lift gate]
Let $B:\^R^4\to\^R^4$ be a fixed linear map. For multisets $\+G,\+H\subseteq\^R^4$, define the edge-lift gate
$$G_B(\+G,\+H)=\{g\odot Bh:g\in\+G,\ h\in\+H\}.$$
\end{definition}

It remains to express the usual tree recursion as admissible gates.  When a
child subtree is attached to its parent, the edge interaction first transforms
the child's two spin-conditioned contributions by summing over the child spin.
The resulting vector is then multiplied coordinatewise with the contribution
already accumulated at the parent.  This is exactly the operation captured by
the above edge-lift gate.

\begin{lemma}\label{lem:admissible-edge-lift}
The gate $G_B$ is admissible with error map $\phi_B(s_{\+G},s_{\+H})=s_{\+G}+s_{\+H}$, and its evaluation time is $T_{G_B}(\mu)=O(\mu)$, where $\mu$ is the total size of the input and output multisets.  In particular, if $|\+G|\le M_{\+G}$ and $|\+H|\le M_{\+H}$, then the gate $G_B$ has output size $M_{G_B}=|G_B(\+G,\+H)|\le M_{\+G}M_{\+H}$ and evaluation time $T_{G_B}(M_{\+G}+M_{\+H}+M_{\+G}M_{\+H})=O(M_{\+G}M_{\+H})$.
\end{lemma}
\begin{proof}
Write $B\+H=\{Bh:h\in\+H\}$.  Applying $B$ preserves zonotope containments and does not change the error weight.  The output of $G_B$ is the coordinatewise product construction applied to $\+G$ and $B\+H$.

For represented zonotopes $A=\+Z(\+G)$ and $C=\+Z(B\+H)$, write
$$
A\star C=\+Z(G_B(\+G,\+H)).
$$
For a vector $g=(g_1,\ldots,g_4)$, let $D_g=\operatorname{diag}(g_1,\ldots,g_4)$, so that $D_gz=g\odot z$ for every $z\in\^R^4$.  Then
\begin{align*}
\sigma_{A\star C}(z) &=\sum_{g\in\+G}\sum_{h\in\+H}\abs{\inner{z}{g\odot Bh}}=\sum_{g\in\+G}\sum_{h\in\+H}\abs{\inner{D_gz}{Bh}}\\
&=\sum_{g\in\+G}\sigma_{\+Z(B\+H)}(D_gz) = \sum_{g\in\+G}\sigma_C(D_gz).
\end{align*}
Since the Hadamard product is commutative, by switching $g$ and $Bh$, we can also obtain that $$\sigma_{A\star C}(z) = \sum_{h\in\+H}\sigma_A(D_{Bh}z).$$

Now suppose the two inputs are represented by compressed multisets $\underline{\+G}$ and $\underline{\+H}$ satisfying
$$
\+Z(\underline{\+G})\subseteq \+Z(\+G)\subseteq(1+\eta)^{s_{\+G}}\+Z(\underline{\+G}),
\qquad
\+Z(\underline{\+H})\subseteq \+Z(\+H)\subseteq(1+\eta)^{s_{\+H}}\+Z(\underline{\+H}).
$$
Set $\underline A=\+Z(\underline{\+G})$, $\underline C=\+Z(B\underline{\+H})$.
Since $B$ is linear, $\underline C\subseteq C\subseteq(1+\eta)^{s_{\+H}}\underline C$.  Therefore, for every $z\in\^R^4$,
\begin{align*}
\sigma_{A\star C}(z) &= \sum_{g\in\+G}\sigma_C(D_gz) \leq (1+\eta)^{s_{\+H}}\sum_{g\in\+G}\sigma_{\underline C}(D_gz)=(1+\eta)^{s_{\+H}}\sigma_{A\star\underline C}(z)\\
&=(1+\eta)^{s_{\+H}}\sum_{\underline h\in\underline{\+H}}\sigma_A(D_{B\underline h}z) \leq(1+\eta)^{s_{\+H}+s_{\+G}}\sum_{\underline h\in\underline{\+H}}\sigma_{\underline A}(D_{B\underline h}z)\\
&=
(1+\eta)^{s_{\+H}+s_{\+G}}\sigma_{\underline A\star\underline C}(z).
\end{align*}
The same monotonicity also gives $\underline A\star\underline C\subseteq A\star C$.  Hence
$$
\+Z(G_B(\underline{\+G},\underline{\+H}))
\subseteq
\+Z(G_B(\+G,\+H))
\subseteq
(1+\eta)^{s_{\+G}+s_{\+H}}
\+Z(G_B(\underline{\+G},\underline{\+H})).
$$
The error map is $\phi_B(s_{\+G},s_{\+H})=s_{\+G}+s_{\+H}$.
The size and time bounds follow by iterating over all pairs $g\in\+G$ and $h\in\+H$: each pair yields one output generator in $O(1)$ time for the fixed dimension $4$, so the total size is upper bounded by $\mu=M_{\+G}+M_{\+H}+M_{\+G}M_{\+H}$, and in particular $T_{G_B}(M_{\+G}+M_{\+H}+M_{\+G}M_{\+H})=O(M_{\+G}M_{\+H})$.
\end{proof}

\begin{proof}[Proof of \Cref{thm:latent-tree}]
We first compute $\gamma=(1/Z_{\^P},1/Z_{\^P},-1/Z_{\^Q},-1/Z_{\^Q})$, which requires the partition functions $Z_{\^P}$ and $Z_{\^Q}$.
Since the underlying graph is a tree, the Ising partition function is
computable exactly by a standard bottom-up dynamic program.  Root the
tree at $r$ and define
\[
F^R_u(s)
=
\exp(h^R_us)
\prod_{v\in \mathrm{ch}(u)}
\sum_{t\in\{\pm1\}}\exp(J^R_{uv}st)F^R_v(t),
\quad \forall s\in\{\pm1\},R\in\{\^P,\^Q\}.
\]
For leaves, $F^R_u(s)=\exp(h^R_us)$.  Then $Z_R=F^R_r(+1)+F^R_r(-1)$.

The family $\+L=\{\+G_u^{(0)}:u\in V\}$ has total size $O(|V|)$.  The full partition functions $Z_{\^P}$ and $Z_{\^Q}$, and hence the vector
$$
\gamma=\left(\frac1{Z_{\^P}},\frac1{Z_{\^P}},-\frac1{Z_{\^Q}},-\frac1{Z_{\^Q}}\right),
$$
are computed by the usual tree partition-function DP in $O(|V|)$ time. Thus the initialization time is $T_0=O(|V|)$.

For a parent-child pair $u,v$ in the rooted orientation of $T$, define the linear map $B_{u,v}:\^R^4\to\^R^4$ by the matrix
$$
B_{u,v}
=
\begin{pmatrix}
\exp({J^{\^P}_{uv}}) & \exp({-J^{\^P}_{uv}}) & 0 & 0\\
\exp({-J^{\^P}_{uv}}) & \exp({J^{\^P}_{uv}}) & 0 & 0\\
0 & 0 & \exp({J^{\^Q}_{uv}}) & \exp({-J^{\^Q}_{uv}})\\
0 & 0 & \exp({-J^{\^Q}_{uv}}) & \exp({J^{\^Q}_{uv}})
\end{pmatrix}.
$$
For each node $u$, let its children in the rooted tree be $v_1,\ldots,v_m$.  After recursively constructing the completed child multisets, apply the edge-lift gates
$$
\+G_u^{(j)}
=G_{B_{u,v_j}}(\+G_u^{(j-1)},\+G_{v_j}),\qquad j=1,\ldots,m,
$$
and set $\+G_u=\+G_u^{(m)}$.  If $u$ is a leaf, then $m=0$ and $\+G_u=\+G_u^{(0)}$.

We prove by induction from the leaves upward that $\+G_u=\{g_u(x_{L_u}):x_{L_u}\in\{\pm1\}^{L_u}\}$.  The base case is exactly the definition of \eqref{eq:G-leaf}: it stores the local external field contribution in the two spin coordinates of each model, and when $u$ is a leaf, the indicator factors enforce the observed symbol $x$.  For the induction step, after attaching $v_1,\ldots,v_j$ we have
$$
\+G_u^{(j)}
=
\left\{
g^{(0)}\odot
\bigodot_{p=1}^j B_{u,v_p}g_{v_p}(x_{L_{v_p}}):
g^{(0)}\in\+G_u^{(0)},\
x_{L_{v_p}}\in\{\pm1\}^{L_{v_p}}\text{ for }p=1,\ldots,j
\right\}.
$$
For $j=m$, fix $R\in\{\^P,\^Q\}$ and $s\in\{\pm1\}$, and consider the coordinate indexed by $(R,s)$.  By the induction hypothesis for every child $v_p$ and definition in \eqref{eq:latent-g},
$$
\left(B_{u,v_p}g_{v_p}(x_{L_{v_p}})\right)_{R,s}
=
\sum_{t\in\{\pm1\}}\exp({J^R_{uv_p}st})Z^{v_p}_R(t,x_{L_{v_p}}).
$$
Therefore the $(R,s)$-coordinate of a generator in $\+G_u$ is
\begin{align*}
  &\exp({h_u^R s})
  \prod_{p=1}^m
  \sum_{t\in\{\pm1\}}\exp({J^R_{uv_p}st})Z^{v_p}_R(t,x_{L_{v_p}}) \notag\\
  =&
  \exp({h_u^R s})
  \prod_{p=1}^m
  \sum_{t\in\{\pm1\}}
  \exp({J^R_{uv_p}st})
  \sum_{\substack{y^{(p)}\in\{\pm1\}^{V_{v_p}}:\\
  y^{(p)}_{v_p}=t,\ y^{(p)}_{L_{v_p}}=x_{L_{v_p}}}}
  \exp\left(
  \sum_{\{i,j\}\in E_{v_p}}J^R_{ij}y^{(p)}_iy^{(p)}_j
  +
  \sum_{i\in V_{v_p}}h_i^R y^{(p)}_i
  \right)\notag \\
  =&
  \exp({h_u^R s})
  \prod_{p=1}^m
  \sum_{\substack{y^{(p)}\in\{\pm1\}^{V_{v_p}}:\\
  y^{(p)}_{L_{v_p}}=x_{L_{v_p}}}}
  \exp\left(
  J^R_{uv_p}s\,y^{(p)}_{v_p}
  +
  \sum_{\{i,j\}\in E_{v_p}}J^R_{ij}y^{(p)}_iy^{(p)}_j
  +
  \sum_{i\in V_{v_p}}h_i^R y^{(p)}_i
  \right)\notag \\
  =&
  \sum_{\substack{y\in\{\pm1\}^{V_u}:\\
  y_u=s,\ y_{L_u}=x_{L_u}}}
  \exp\left(
  \sum_{\{i,j\}\in E_u}J^R_{ij}y_iy_j
  +
  \sum_{i\in V_u}h_i^R y_i
  \right)
  =Z^u_R(s,x_{L_u}).
\end{align*}
The first and last equality come from the definition in \eqref{eq:conditional-partition-function}, and the third equality uses $V_u = \{u\}\cup V_{v_1}\cup\cdots\cup V_{v_m}$ and $E_u = \{\{u,v_p\}:p=1,\ldots,m\}\cup E_{v_1}\cup\cdots\cup E_{v_m}$. Thus by induction we have $\+G_r=\{g_r(x_{L}):x_{L}\in\{\pm1\}^{L}\} = \+G^{\mathrm{root}}$.

The gate sequence is directed acyclic and has one edge-lift gate per edge, so $N=|E|=|V|-1$.
All gates are admissible by \Cref{lem:admissible-edge-lift}.  Since $\phi_B(s_{\+G},s_{\+H})=s_{\+G}+s_{\+H}$, and the abstract algorithm adds one compression after every gate, the same bottom-up induction shows that the completed multiset at node $u$ has error weight equal to the number of edges in the subtree rooted at $u$.  Therefore the final error weight is $S=|E|=|V|-1$.
Applying \Cref{thm:abstract-fptas} gives the stated multiplicative approximation.

It remains to substitute the parameters into the running-time bound.  Algorithm \ref{alg:abstract-FPTAS} uses
$\eta=(1+\varepsilon)^{1/S}-1$, 
so $\eta^{-1}=O(|V|/\varepsilon)$.  By \Cref{lem:compression}, with $K=4$, every stored compressed multiset has size
$$ M=O((|V|/\varepsilon)^4). $$
For each edge-lift gate $G_B$, $M_{G_B}\le M^2$, and $T_{G_B}(2M+M^2)=O(M^2)$.  Since there are $|V|-1$ gates and $S=|V|-1$, \Cref{thm:abstract-fptas} gives total running time
\begin{equation*}
O\left(|V|\left(M^2+M^2(|V|/\varepsilon)^4\right)\right)
=O\left(|V|^{13}\varepsilon^{-12}\right).\qedhere
\end{equation*}
\end{proof}

\section*{Acknowledgements}
The author is grateful to Weiming Feng for introducing the subject of mixtures of Markov chains during the development of the coupling-based approach in~\cite{FengFYZ26}, for carefully verifying the correctness of the results, and for valuable guidance in the preparation of this manuscript. The author also thanks the anonymous reviewers for their constructive comments and suggestions, which have been incorporated into the present revision.

For language editing of this manuscript, including grammar, style, and wording polishing, the author used GPT-5.5 Pro for the initial version and GPT-5.6 Sol and GPT-6 Astra for subsequent revisions.

\bibliographystyle{alpha}
\bibliography{refs}

@inproceedings{DaganKD22,
  author       = {Yuval Dagan and
                  Anthimos Vardis Kandiros and
                  Constantinos Daskalakis},
  title        = {{EM}'s Convergence in {Gaussian} Latent Tree Models},
  booktitle    = {COLT},
  series       = {Proceedings of Machine Learning Research},
  volume       = {178},
  pages        = {2597--2667},
  publisher    = {{PMLR}},
  year         = {2022},
}

@inproceedings{SpaehT23,
  author       = {Fabian Spaeh and
                  Charalampos E. Tsourakakis},
  title        = {Learning Mixtures of {Markov Chains} with Quality Guarantees},
  booktitle    = {{WWW}},
  pages        = {662--672},
  publisher    = {{ACM}},
  year         = {2023},
}

@inproceedings{GuptaKV16,
  author       = {Rishi Gupta and
                  Ravi Kumar and
                  Sergei Vassilvitskii},
  title        = {On Mixtures of {Markov Chains}},
  booktitle    = {NIPS},
  pages        = {3441--3449},
  year         = {2016},
}

@article{FengFYZ26,
  author       = {Weiming Feng and
                  Yucheng Fu and
                  Minji Yang and
                  Anqi Zhang},
  title        = {On Computing Total Variation Distance Between Mixtures of Product
                  Distributions},
  journal      = {CoRR},
  volume       = {abs/2605.03839},
  year         = {2026},
}

@article{BCFM26,
  author       = {Arnab Bhattacharyya and
                  Graham Cormode and
                  Yucheng Fu and
                  Kuldeep S. Meel},
  title        = {Total Variation Distance Estimation through Domain Reduction},
  journal      = {CoRR},
  volume       = {abs/2609.18707},
  year         = {2026},
  doi          = {10.48550/arXiv.2609.18707},
  url          = {https://arxiv.org/abs/2609.18707},
}

@article{BreslerK20,
  author       = {Bresler, Guy and Karzand, Mina},
  title        = {Learning a tree-structured {Ising} model in order to make predictions},
  journal      = {Ann. Statist.},
  volume       = {48},
  number       = {2},
  pages        = {713--737},
  year         = {2020},
  doi          = {10.1214/19-AOS1808},
}

@inproceedings{CohenP15,
  author       = {Michael B. Cohen and
                  Richard Peng},
  title        = {{$\ell_p$} Row Sampling by {Lewis} Weights},
  booktitle    = {{STOC}},
  pages        = {183--192},
  publisher    = {{ACM}},
  year         = {2015},
}

@article{BourgainLM89,
  author       = {Bourgain, J. and Lindenstrauss, J. and Milman, V. D.},
  title        = {Approximation of zonoids by zonotopes},
  journal      = {Acta Math.},
  volume       = {162},
  number       = {1-2},
  pages        = {73--141},
  year         = {1989},
  doi          = {10.1007/BF02392835},
}

@article{Matousek96Zonotopes,
  author       = {Matou{\v{s}}ek, Ji{\v{r}}{\'{\i}}},
  title        = {Improved upper bounds for approximation by zonotopes},
  journal      = {Acta Math.},
  volume       = {177},
  number       = {1},
  pages        = {55--73},
  year         = {1996},
  doi          = {10.1007/BF02392598},
}

@article{CampiHW94,
  author       = {Campi, S. and Haas, D. and Weil, W.},
  title        = {Approximation of zonoids by zonotopes in fixed directions},
  journal      = {Discrete Comput. Geom.},
  volume       = {11},
  number       = {4},
  pages        = {419--431},
  year         = {1994},
  doi          = {10.1007/BF02574016},
}

@article{YangO22,
  author       = {Liren Yang and
                  Necmiye Ozay},
  title        = {Scalable Zonotopic Under-Approximation of Backward Reachable
                  Sets for Uncertain Linear Systems},
  journal      = {{IEEE} Control. Syst. Lett.},
  volume       = {6},
  pages        = {1555--1560},
  year         = {2022},
}

@article{LutzowKA25,
  author       = {Laura L{\"{u}}tzow and
                  Niklas Kochdumper and
                  Matthias Althoff},
  title        = {Underapproximative Methods for the Order Reduction of Zonotopes},
  journal      = {{IEEE} Control. Syst. Lett.},
  volume       = {9},
  pages        = {1730--1735},
  year         = {2025},
}

@article{Lewis78,
  author       = {Lewis, D. R.},
  title        = {Finite dimensional subspaces of {$L_p$}},
  journal      = {Studia Math.},
  volume       = {63},
  number       = {2},
  pages        = {207--212},
  year         = {1978},
}

@article{ReisRothvoss26,
  author       = {Victor Reis and
                  Thomas Rothvoss},
  title        = {Linear-size {$\ell_1$} sparsifiers},
  journal      = {CoRR},
  volume       = {abs/2606.28147},
  year         = {2026},
}

@inproceedings{KandirosDDC23,
  author       = {Anthimos Vardis Kandiros and
                  Constantinos Daskalakis and
                  Yuval Dagan and
                  Davin Choo},
  title        = {Learning and Testing Latent-Tree {Ising} Models Efficiently},
  booktitle    = {COLT},
  series       = {Proceedings of Machine Learning Research},
  volume       = {195},
  pages        = {1666--1729},
  publisher    = {{PMLR}},
  year         = {2023},
}

@inproceedings{KausikTT23,
  author       = {Chinmaya Kausik and
                  Kevin Tan and
                  Ambuj Tewari},
  title        = {Learning Mixtures of {Markov Chains and MDPs}},
  booktitle    = {ICML},
  series       = {Proceedings of Machine Learning Research},
  volume       = {202},
  pages        = {15970--16017},
  publisher    = {{PMLR}},
  year         = {2023},
}

@inproceedings{AshtianiBHLPM18,
  author       = {Hassan Ashtiani and
                  Shai Ben{-}David and
                  Nicholas J. A. Harvey and
                  Christopher Liaw and
                  Abbas Mehrabian and
                  Yaniv Plan},
  title        = {Nearly tight sample complexity bounds for learning {mixtures of Gaussians}
                  via sample compression schemes},
  booktitle    = {NeurIPS},
  pages        = {3416--3425},
  year         = {2018},
}

@inproceedings{LiS17,
  author       = {Jerry Li and
                  Ludwig Schmidt},
  title        = {Robust and Proper Learning for Mixtures of {Gaussians} via Systems of
                  Polynomial Inequalities},
  booktitle    = {COLT},
  series       = {Proceedings of Machine Learning Research},
  volume       = {65},
  pages        = {1302--1382},
  publisher    = {{PMLR}},
  year         = {2017},
}

@inproceedings{JainO14,
  author       = {Prateek Jain and
                  Sewoong Oh},
  title        = {Learning Mixtures of Discrete Product Distributions using Spectral
                  Decompositions},
  booktitle    = {COLT},
  series       = {{JMLR} Workshop and Conference Proceedings},
  volume       = {35},
  pages        = {824--856},
  publisher    = {JMLR.org},
  year         = {2014},
}

@inproceedings{GordonJMRS24,
  author       = {Spencer L. Gordon and
                  Erik Jahn and
                  Bijan Mazaheri and
                  Yuval Rabani and
                  Leonard J. Schulman},
  title        = {Identification of mixtures of discrete product distributions in near-optimal
                  sample and time complexity},
  booktitle    = {COLT},
  series       = {Proceedings of Machine Learning Research},
  volume       = {247},
  pages        = {2071--2091},
  publisher    = {{PMLR}},
  year         = {2024},
}

@inproceedings{GordonMRS21,
  author       = {Spencer Gordon and
                  Bijan H. Mazaheri and
                  Yuval Rabani and
                  Leonard J. Schulman},
  title        = {Source Identification for Mixtures of Product Distributions},
  booktitle    = {COLT},
  series       = {Proceedings of Machine Learning Research},
  volume       = {134},
  pages        = {2193--2216},
  publisher    = {{PMLR}},
  year         = {2021},
}

@article{FeldmanOS08,
  author       = {Jon Feldman and
                  Ryan O'Donnell and
                  Rocco A. Servedio},
  title        = {Learning Mixtures of Product Distributions over Discrete Domains},
  journal      = {{SIAM} J. Comput.},
  volume       = {37},
  number       = {5},
  pages        = {1536--1564},
  year         = {2008},
}

@inproceedings{ChenM19,
  author       = {Sitan Chen and
                  Ankur Moitra},
  title        = {Beyond the low-degree algorithm: mixtures of subcubes and their applications},
  booktitle    = {STOC},
  pages        = {869--880},
  publisher    = {{ACM}},
  year         = {2019},
}

@inproceedings{0001GMMPV25,
  author       = {Arnab Bhattacharyya and
                  Sutanu Gayen and
                  Kuldeep S. Meel and
                  Dimitrios Myrisiotis and
                  Aduri Pavan and
                  N. V. Vinodchandran},
  title        = {Computational Explorations of Total Variation Distance},
  booktitle    = {ICLR},
  publisher    = {OpenReview.net},
  year         = {2025},
}

@inproceedings{KwonECM23,
  author       = {Jeongyeol Kwon and
                  Yonathan Efroni and
                  Constantine Caramanis and
                  Shie Mannor},
  title        = {{Reward-Mixing MDPs} with Few Latent Contexts are Learnable},
  booktitle    = {ICML},
  series       = {Proceedings of Machine Learning Research},
  volume       = {202},
  pages        = {18057--18082},
  publisher    = {{PMLR}},
  year         = {2023},
}

@inproceedings{BFS25,
  author       = {Arnab Bhattacharyya and
                  Weiming Feng and
                  Piyush Srivastava},
  title        = {Approximating the Total Variation Distance between {G}aussians},
  booktitle    = {{AISTATS}},
  series       = {Proceedings of Machine Learning Research},
  volume       = {258},
  pages        = {1846--1854},
  publisher    = {{PMLR}},
  year         = {2025},
}

@inproceedings{feng2025approximating,
   author       = {Weiming Feng and
                  Hongyang Liu and
                  Minji Yang},
  title        = {Approximating the total variation distance between spin systems},
  booktitle    = {COLT},
  series       = {Proceedings of Machine Learning Research},
  volume       = {291},
  pages        = {1974--2025},
  publisher    = {{PMLR}},
  year         = {2025},
}

@inproceedings{CanonneR14,
  author       = {Cl{\'{e}}ment L. Canonne and
                  Ronitt Rubinfeld},
  title        = {Testing Probability Distributions Underlying Aggregated Data},
  booktitle    = {{ICALP} {(1)}},
  series       = {Lecture Notes in Computer Science},
  volume       = {8572},
  pages        = {283--295},
  publisher    = {Springer},
  year         = {2014},
}

@inproceedings{Kiefer18,
  author       = {Stefan Kiefer},
  title        = {On Computing the Total Variation Distance of Hidden {M}arkov Models},
  booktitle    = {ICALP},
  series       = {LIPIcs},
  volume       = {107},
  pages        = {130:1--130:13},
  publisher    = {Schloss Dagstuhl - Leibniz-Zentrum f{\"{u}}r Informatik},
  year         = {2018},
}

@inproceedings{ChenK14,
  author       = {Taolue Chen and
                  Stefan Kiefer},
  title        = {On the Total Variation Distance of Labelled {Markov} Chains},
  booktitle    = {{CSL-LICS}},
  pages        = {33:1--33:10},
  publisher    = {{ACM}},
  year         = {2014},
}

@inproceedings{FengLL24,
  author       = {Weiming Feng and
                  Liqiang Liu and
                  Tianren Liu},
  title        = {On Deterministically Approximating Total Variation Distance},
  booktitle    = {SODA},
  pages        = {1766--1791},
  publisher    = {{SIAM}},
  year         = {2024},
}

@inproceedings{BGMMPV23,
  author       = {Arnab Bhattacharyya and
                  Sutanu Gayen and
                  Kuldeep S. Meel and
                  Dimitrios Myrisiotis and
                  A. Pavan and
                  N. V. Vinodchandran},
  title        = {On Approximating Total Variation Distance},
  booktitle    = {IJCAI},
  pages        = {3479--3487},
  publisher    = {ijcai.org},
  year         = {2023},
}

@article{sahai2003complete,
  author       = {Amit Sahai and
                  Salil P. Vadhan},
  title        = {A Complete Problem for Statistical Zero Knowledge},
  journal      = {J. {ACM}},
  volume       = {50},
  number       = {2},
  pages        = {196--249},
  year         = {2003},
}

@article{FGJW23,
  author       = {Weiming Feng and
                  Heng Guo and
                  Mark Jerrum and
                  Jiaheng Wang},
  title        = {A simple polynomial-time approximation algorithm for the total variation
                  distance between two product distributions},
  journal      = {TheoretiCS},
  volume       = {2},
  eid          = {7},
  year         = {2023},
}

@inproceedings{BGMMPV24ICML,
  author       = {Arnab Bhattacharyya and
                  Sutanu Gayen and
                  Kuldeep S. Meel and
                  Dimitrios Myrisiotis and
                  A. Pavan and
                  N. V. Vinodchandran},
  title        = {Total Variation Distance Meets Probabilistic Inference},
  booktitle    = {{ICML}},
  series       = {Proceedings of Machine Learning Research},
  volume       = {235},
  pages        = {3776--3794},
  publisher    = {{PMLR}},
  year         = {2024},
}

@inproceedings{BGMV20,
  author       = {Arnab Bhattacharyya and
                  Sutanu Gayen and
                  Kuldeep S. Meel and
                  N. V. Vinodchandran},
  title        = {Efficient Distance Approximation for Structured High-Dimensional Distributions
                  via Learning},
  booktitle    = {NeurIPS},
  year         = {2020},
}

\appendix
\section{Proof of \texorpdfstring{\Cref{lem:lewis-rounding}}{Lemma 3.4}}\label{sec:appendix}

To give an algorithm for \Cref{lem:lewis-rounding}, we introduce the $\ell_1$ Lewis weights, following \cite[Definition 2.2]{CohenP15}.
\begin{definition}[\texorpdfstring{$\ell_1$}{l1} Lewis weights]\label{def:l1-lewis}
Let $G\in\mathbb R^{m\times r}$ be a matrix of full column rank, with nonzero rows $g_1^T, \ldots, g_m^T$.  The $\ell_1$ Lewis weights of $G$ are the entries of the positive vector $w = (w_1,\ldots,w_m)$ satisfying
$$ w_i^2=g_i^T\left(\sum_{j=1}^m\frac{1}{w_j}g_jg_j^T\right)^{-1}g_i, \qquad \forall i\in [m]. $$
\end{definition}

\begin{lemma}\label{lem:lewis-ellipsoid-comparison}
Let $G\in\mathbb R^{m\times r}$ be a matrix of full column rank and let $w = (w_1,\ldots,w_m)$ be the vector of $\ell_1$ Lewis weights of $G$. Set $ M=\sum_{i=1}^m\frac1{w_i}g_ig_i^T$.
Then, for every $z\in\mathbb R^r$,
$$
\sqrt{z^TMz}\leq \|Gz\|_1\leq \sqrt r\,\sqrt{z^TMz}.
$$

Moreover, if $\alpha\ge1$ and $\widetilde w$ is an $\alpha$-approximation to the $\ell_1$ Lewis weights, that is 
$$ \frac{1}{\alpha}w_i\le\widetilde w_i\le\alpha w_i, \qquad \forall i\in [m], $$
and $\widetilde M=\sum_{i=1}^m\frac{1}{\widetilde w_i}g_ig_i^T$, then for every $z\in\mathbb R^r$
$$
\frac1{\sqrt\alpha}\sqrt{z^T\widetilde Mz}
\le \|Gz\|_1
\le \sqrt{\alpha r}\sqrt{z^T\widetilde Mz}.
$$
\end{lemma}

\begin{proof}
Let $W=\operatorname{diag}(w_1,\ldots,w_m)$ and $ B=W^{-1/2}GM^{-1/2}$. Then we have 
$$ B^TB=M^{-1/2}\cdot (G^TW^{-1}G)\cdot M^{-1/2}=M^{-1/2}MM^{-1/2} = I_r. $$
Let $b_i^T$ be the $i$-th row of $B$. Then by \Cref{def:l1-lewis},
$$
\|b_i\|_2^2=\frac1{w_i}g_i^TM^{-1}g_i=w_i.
$$
In particular, the sum-of-leverage-scores identity \cite[Fact 2.1]{CohenP15} gives
\begin{align}  \label{eq:l1-lewis-sum}
  \sum_{i=1}^m w_i=\sum_{i=1}^m\|b_i\|_2^2=\operatorname{tr}(B^TB)=r.
\end{align}

For $z\in\mathbb R^r$, set $v=M^{1/2}z$. Then
$$
\|Gz\|_1=\sum_{i=1}^m |\inner{g_i}{z}|
=\sum_{i=1}^m\sqrt{w_i}\,|\inner{b_i}{v}|,
\text{ and }
\sqrt{z^TMz}=\|v\|_2.
$$
Since $B^TB=I_r$,
$$
\sum_{i=1}^m|\inner{b_i}{v}|^2
=\|Bv\|_2^2
=v^TB^TBv
=\|v\|_2^2.
$$

The upper bound follows from~\eqref{eq:l1-lewis-sum} and the Cauchy--Schwarz inequality:
$$
\sum_{i=1}^m\sqrt{w_i}\,|\inner{b_i}{v}|
\le\left(\sum_{i=1}^m w_i\right)^{1/2}
\left(\sum_{i=1}^m|\inner{b_i}{v}|^2\right)^{1/2}
=\sqrt r\,\|v\|_2.
$$
For the lower bound, the claim is trivial when $v=0$. Assume $v\ne0$. Then
$$
|\inner{b_i}{v}|\le\|b_i\|_2\|v\|_2=\sqrt{w_i}\|v\|_2,
$$
and hence
$$
\sum_{i=1}^m\sqrt{w_i}\,|\inner{b_i}{v}|
\ge\frac1{\|v\|_2}\sum_{i=1}^m|\inner{b_i}{v}|^2
=\|v\|_2.
$$
This proves the exact comparison for $\ell_1$ Lewis weights.

For the approximate statement, the coordinatewise bounds on $\widetilde w$ imply
$$ \alpha^{-1}M\preceq\widetilde M\preceq\alpha M. $$
Combining these semidefinite inequalities with the exact comparison gives
\begin{equation*}\frac{1}{\sqrt\alpha}\sqrt{z^T\widetilde Mz} \leq \sqrt{z^TMz} \leq \|Gz\|_1 \leq \sqrt{r}\sqrt{z^TMz} \leq \sqrt{\alpha r}\sqrt{z^T\widetilde Mz}. \qedhere\end{equation*}
\end{proof}

We therefore adapt the $\ell_1$ specialization of the algorithm given in \cite[Section 3]{CohenP15}. All logarithms in this appendix are natural. Algorithm \ref{alg:cohen-peng-l1-lewis} returns an $m^\theta$-approximation to the $\ell_1$ Lewis weights of $G$.  In each iteration, forming $M^{(t)}$ costs $O(mr^2)$ arithmetic operations, inverting $M^{(t)}$ costs $O(r^3)$, and computing all quadratic forms $g_i^T(M^{(t)})^{-1}g_i$ costs $O(mr^2)$; hence the algorithm uses $O((mr^2+r^3)T)$ arithmetic operations.
\begin{algorithm}[ht]
    \caption{Cohen--Peng algorithm for approximate \texorpdfstring{$\ell_1$}{l1} Lewis weights}
    \label{alg:cohen-peng-l1-lewis}
    \SetKwInOut{Input}{Input}
    \SetKwInOut{Output}{Output}
    \underline{\textsc{ApproxLewisWeights}}$_1(G,\theta,T)$\;
    \Input{Matrix $G \in \^R^{m\times r}$ with rows $g_i^T$, a target exponent $\theta\in(0,1]$, and an iteration count $T\ge2\log(2/\theta)$}
    \Output{$m^\theta$-approximation $\widetilde w=(\widetilde w_i)_{i=1}^m$ to the $\ell_1$ Lewis weights of $G$;}
    Initialize $w_i^{(0)}\gets1$ for all $i\in[m]$\;
    \For{$t=0,\ldots,T-1$}{
        Let $W^{(t)}\gets\operatorname{diag}(w_1^{(t)},\ldots,w_m^{(t)})$\;
        Compute $M^{(t)}\gets \sum_{j=1}^m\frac{1}{w_j^{(t)}}g_jg_j^T$\;
        Update $w_i^{(t+1)}\gets\sqrt{g_i^T(M^{(t)})^{-1}g_i}$ for all $i \in [m]$\;  
    }
    \Return $\widetilde w=(w_i^{(T)})_{i=1}^m$\;
\end{algorithm}

\paragraph{Initialization and iteration guarantee}
We verify the guarantee from the all-ones initialization, including the first-update bound of \cite[Lemma 3.5]{CohenP15}. Let $G$ have full column rank and nonzero rows, let $w^*$ be its exact $\ell_1$ Lewis weights, and set
$$
M^*=\sum_{i=1}^m\frac{1}{w_i^*}g_ig_i^T,
\qquad H=G^TG=M^{(0)}.
$$
For the exact weights, the matrix $B$ in the proof of \Cref{lem:lewis-ellipsoid-comparison} has orthonormal columns and squared row norms $w_i^*$. Thus $BB^T$ is an orthogonal projection and $0<w_i^*\le1$. In particular, $H\preceq M^*$. On the other hand, \Cref{lem:lewis-ellipsoid-comparison} and the Cauchy--Schwarz inequality give, for every $z\in\mathbb R^r$,
$$
z^TM^*z\le\|Gz\|_1^2\le m\|Gz\|_2^2=mz^THz.
$$
Consequently,
$$
\frac1mM^*\preceq H\preceq M^*,
\qquad (M^*)^{-1}\preceq H^{-1}\preceq m(M^*)^{-1}.
$$
The first exact update from $w^{(0)}=\mathbf1$ therefore satisfies
\begin{equation}\label{eq:lewis-first-update}
w_i^*\le w_i^{(1)}=\sqrt{g_i^TH^{-1}g_i}\le\sqrt m\,w_i^*,
\qquad i\in[m].
\end{equation}
This bound applies after the first update; the all-ones vector itself need not be a multiplicative approximation with a factor depending only on $m$.

For subsequent updates, suppose $w^*\le w^{(t)}\le\alpha w^*$ coordinatewise for some $t\ge1$. Then $\alpha^{-1}M^*\preceq M^{(t)}\preceq M^*$. Inverting these inequalities shows that the next iterate satisfies
$$
w_i^*\le w_i^{(t+1)}=\sqrt{g_i^T(M^{(t)})^{-1}g_i}\le\sqrt\alpha\,w_i^*.
$$
This is the exact-update contraction in \cite[Corollary 3.3]{CohenP15}. Starting with \eqref{eq:lewis-first-update}, induction gives
\begin{equation}\label{eq:lewis-iteration-bound}
w_i^*\le w_i^{(t)}\le m^{2^{-t}}w_i^*,
\qquad t\ge1,\quad i\in[m].
\end{equation}
For $\theta\in(0,1]$ and the stated $T\ge2\log(2/\theta)$, we have $T\ge1$ and, since $\log2\ge1/2$,
$$
2^{-T}\le e^{-T/2}\le\theta/2.
$$
Thus the output has approximation factor at most $m^{\theta/2}\le m^\theta$, as required.

\begin{proof}[Proof of \cref{lem:lewis-rounding}]
Delete zero generators.  If $r=0$, then $\+Z=\{0\}$ and the lemma holds trivially.  Otherwise choose a linear isomorphism $U:E\to\mathbb R^r$, and write
$$
\widetilde g_i=Ug_i\in\mathbb R^r,\qquad
\widetilde{\+Z}=U\+Z=\sum_{i=1}^m[-\widetilde g_i,\widetilde g_i].
$$
Let $\widetilde G\in\mathbb R^{m\times r}$ have $i$-th row $\widetilde g_i^T$.  Since the original generators span $E$, the rows of $\widetilde G$ span $\mathbb R^r$, so $\widetilde G$ has full column rank.  If $m=1$, set $\widetilde w_1=1$, which is exact.  Otherwise, take $\theta=\log2/\log m$ and run Algorithm \ref{alg:cohen-peng-l1-lewis} with input $(\widetilde G,\theta,T)$, where
$$ T=\left\lceil2\log\left(\frac{2\log (m+1)}{\log2}\right)\right\rceil. $$
For $m\ge2$, this choice satisfies
$$ T\geq2\log\left(\frac{2\log m}{\log2}\right)=2\log(2/\theta),$$
so \eqref{eq:lewis-iteration-bound} and the bound $2^{-T}\le\theta/2$ give an approximation factor at most
$$
m^{2^{-T}}\le m^{\theta/2}=\sqrt2\le2.
$$
Hence the resulting weights give a $2$-approximation to the $\ell_1$ Lewis weights of $\widetilde G$. The case $m=2$ has $\theta=1$, which is included in the algorithm's allowed range; $m=1$ was handled exactly above. In both cases, let $\widetilde w_i$ denote the chosen weights and set
$$ \widetilde M=\sum_{i=1}^m\frac{1}{\widetilde w_i}\widetilde g_i\widetilde g_i^T, \qquad \widetilde A=\sqrt2\,\widetilde M^{-1/2}. $$
Since $\sigma_{\widetilde{\+Z}}(z) = \sum_{i=1}^m\abs{\inner{\widetilde g_i}{z}} = \|\widetilde Gz\|_1$, \Cref{lem:lewis-ellipsoid-comparison} with $\alpha=2$ gives that, for every $z\in\mathbb R^r$,
$$ \frac1{\sqrt2}\sqrt{z^T\widetilde Mz}\le\sigma_{\widetilde{\+Z}}(z)\le\sqrt{2r}\sqrt{z^T\widetilde Mz}. $$
For a unit vector $u\in\mathbb R^r$, $\sigma_{\widetilde A\widetilde{\+Z}}(u)=\sigma_{\widetilde{\+Z}}(\widetilde A^Tu)$, and $(\widetilde A^Tu)^T\widetilde M(\widetilde A^Tu)=2$.  Hence
$$ 1\le\sigma_{\widetilde A\widetilde{\+Z}}(u)\le2\sqrt r $$
for every unit $u$, which is equivalent to $B_2^r\subseteq \widetilde A\widetilde{\+Z}\subseteq2\sqrt r\,B_2^r$.  Finally define the map in the statement by
$$
A=\widetilde A U:E\to\mathbb R^r.
$$
Then $A\+Z=\widetilde A\widetilde{\+Z}$, so $B_2^r\subseteq A\+Z\subseteq2\sqrt r\,B_2^r$.

Finally, after the weights are computed, forming $\widetilde M$ and computing $\widetilde M^{-1/2}$ cost another $O(mr^2+r^3)$ arithmetic operations.  Since $T=O(\log\log(m+2))$, the total time, including the construction of $U$ and conversion to $r$-dimensional coordinates, is $O_K(m\log\log(m+2))$.
\end{proof}
\end{document}